\documentclass[12pt]{article}
\usepackage{graphics}
\usepackage{amsmath}
\usepackage{amssymb}
\usepackage{graphicx}
\usepackage{makeidx}
\usepackage{mathrsfs}
\usepackage{amsfonts}
\usepackage[mathscr]{eucal}
\usepackage{amsmath}
\usepackage{float}

\DeclareGraphicsExtensions{.eps,.pdf,.jpg,.png}

\def\RR{\vbox {\hbox to 8.9pt {I\hskip-2.1pt R\hfil}}}

\def\pni{\par\noindent}
\def\vsh{\smallskip}

\def\vsp{\vsh\pni} 

\newtheorem{theorem}{Theorem}

\newtheorem{lemma}[theorem]{Lemma}

\newtheorem{remark}[theorem]{Remark}

\newenvironment{proof}[1][Proof]{\noindent\textbf{#1.} }{\ \rule{0.5em}{0.5em}}
\numberwithin{equation}{section}
\newcommand{\realpart}[1]{\text{Re}\left({#1}\right)}
\newcommand{\imagpart}[1]{\text{Im}\left({#1}\right)}
\newcommand{\bigO}[1]{\mathcal{O}\left(#1\right)}
\renewcommand{\leq}{\leqslant}
\renewcommand{\geq}{\geqslant}
\begin{document}
\font\title=cmbx12 scaled\magstep2
\font\bfs=cmbx12 scaled\magstep1
\font\little=cmr10
\begin{center}
{\title Delta-pulse solution \\ in the  Zener viscoelastic model}
 \\  [0.10truecm]
 Andrea MENTRELLI$^{(1), (2), (3)}$ 
 \\
 Juan Luis GONZ{\'A}LEZ  SANTANDER$^{(4)}$,
 \\ Francesco MAINARDI$^{(5)}$, 
\\[0.25truecm]
$^{(1)}$ Department of Mathematics, University of Bologna Piazza di Porta San Donato 5, Bologna, 4016. Italy;
{andrea.mentrelli@unibo.it}
\\
$^{(2)}$ Alma Mater Research Center of Applied Mathematics AM$^2$,
  University of Bologna.  Via Saragozza 8,  Bologna, 40123 Italy;
\\ 
$^{(3)}$ Section of Bologna, I.S. FLAG, Italian National Institute for    Nuclear
Physics  (I.N.F.N.),  Viale Berti Pichat 6/2, Bologna, 40127 Bologna, Italy;
\\
$^{(4)}$  Department of Mathematics, University of  Oviedo.
  C/ Leopoldo Calvo Sotelo 18,  Oviedo 33007, Asturias, Spain;
 {gonzalezmarjuan@uniovi.es}
\\
$^{(5)}$  
 Department of Physics and Astronomy, University of Bologna and INFN,
 Via Irnerio 46,  Bologna 40126, Italy;
\\{ francesco.mainardi@unibo.it; mainardi@bo.infn.it, fracalmo@gmail.com}

 \vskip 0.25truecm
  {\bf Published on line in
 Journal of Elasticity, 2 Dec. 2026, 158:63}
 \\
 {\bf DOI:   10.1007/s10659-026-10239-8} 
   \end{center}
\begin{abstract}
	We derive an exact integral representation of the response to an impulsive boundary excitation in a semi-infinite, homogeneous, linear viscoelastic medium governed by the Zener model. Starting from the Bromwich integral, we analytically invert the relevant Laplace transforms and express the response in terms of modified Bessel functions of the first kind and Macdonald functions of half-integer order. The representation reduces to the known Maxwell solution in the appropriate limit. As an independent computational approach, we characterize the saddle points of the phase function and construct the associated steepest-descent path (SDP). Its branches satisfy a sixth-degree polynomial equation in the imaginary part of the complex variable, which is cubic in its square. Numerical comparisons for several model parameters confirm the agreement of the two methods. The integral representation reveals the structure of the response, whereas the SDP approach requires no closed-form inverse Laplace transform and can therefore be applied to more general viscoelastic models.
 \vsp
\vsp {\bf Mathematics Subject Classification}:
{74D05, 44A10, 33C10}

\end{abstract}


\section{Introduction} \label{Sect: Introduction}

We compare two methods for computing the delta-pulse response of the Zener model in linear viscoelasticity. Also known as the standard linear solid (SLS), this model is widely used in Earth rheology. Its creep and relaxation functions involve exponential terms with retardation time $\tau_\epsilon$ and relaxation time $\tau_\sigma$, respectively, where $0<\tau_\sigma<\tau_\epsilon<\infty$. In the limit of infinite retardation time, taken at fixed relaxation time and instantaneous modulus, the SLS reduces to the classical Maxwell model, whose creep compliance grows linearly with time. Both models are used to describe transient waves and are among the simplest viscoelastic models with a finite wave-front velocity.

Earlier investigations of transient waves in viscoelastic media have used several complementary techniques. The problem was studied by Jeffreys \cite{Jeffreys} and later by Mainardi \cite{Mainardi1972}. Morrison \cite{Morrison} and Chin \cite{Chin} used transform methods; Chin also discussed the characteristic and subcharacteristic structure of the governing equations. Buchen and Mainardi \cite{Buchen1975} developed asymptotic expansions, while Mainardi and Turchetti \cite{Mainardi1975} used Pad\'e approximants. More recently, Colombaro et al.\ \cite{Colombaro} revisited and generalized the Buchen--Mainardi approach, developing systematic wave-front and long-time asymptotic expansions for several ordinary and fractional viscoelastic models. These approaches describe the local behavior of transient waves, particularly near the wave front. Here, we instead derive an exact integral representation valid throughout the region behind the Zener wave front. In our recent paper \cite{MainardiMentrelliSantander}, we introduced a steepest-descent approach for the Klein--Gordon equation with dissipation, including the Maxwell (telegraph) model, whose inverse Laplace transform is known explicitly in terms of modified Bessel functions. We extend that approach to the Zener model, whose two finite characteristic times produce a different analytic structure: two branch points joined by a finite branch cut and, for $0<\mu<1$, four saddle points on the associated two-sheeted Riemann surface. We also derive an independent convolution representation of the impulse response, allowing the two methods to be compared. Both methods concern the inversion of the Laplace transform of the delta-pulse response.
\newpage
Sections~\ref{Section 2: Essentials} and \ref{Section 3: Wave equations} recall the relevant elements of linear viscoelasticity, the wave equations, and the complex index of refraction. Section~\ref{Section 4: Integral representation} derives an exact integral representation of the Zener impulse response by analytical inversion and convolution. Section~\ref{Section 5: SDP} evaluates the same Bromwich integral by deforming the contour onto the SDP. This construction is inspired by the classical analysis of Sommerfeld and Brillouin for dispersive electromagnetic waves \cite{Brillouin}; here, however, we integrate the complete contour rather than retaining only a local saddle-point approximation. Section~\ref{Section 6: Conclusions} presents the conclusions.

\section{Essentials of linear viscoelasticity} \label{Section 2: Essentials}

In linear viscoelasticity, the stress at a material point depends on the entire strain history at that point. A viscoelastic body can be viewed as a linear system with stress as the input and strain as the output, or conversely. We consider the one-dimensional case, with spatial coordinate $x$ and time $t$.
Throughout the paper, $r(x,t)$ denotes the field variable whose propagation is considered: the stress $\sigma$, strain $\epsilon$, particle displacement $u$, or particle velocity $v$. In the present one-dimensional setting, each variable satisfies the same propagation equation under the corresponding boundary excitation. The physical meaning of $r$ is determined by the quantity prescribed at the boundary.

A linearly viscoelastic material combines elastic and viscous behavior, represented mechanically by springs governed by Hooke's law and dashpots governed by Newton's law. For a homogeneous body of density $\rho$, these laws read
\begin{equation}
	\text{Hooke's law:} \; \sigma(x,t) = m \, \epsilon(x,t), \quad
	\text{Newton's law:} \; \sigma(x,t) = \eta \, \frac{\partial \epsilon}{\partial t} (x,t),
	\label{eq:HOOKE-NEWTON}
\end{equation}
where $m$ denotes the elastic modulus and $\eta$ the dynamic viscosity.

The constitutive law of a linear viscoelastic material can be expressed through hereditary integrals involving the creep compliance $J(t)$ and relaxation modulus $G(t)$. The function $J(t)$ is the strain produced by a unit step in stress, whereas $G(t)$ is the stress produced by a unit step in strain.

The medium is initially at rest for $t<0$. We denote the Laplace transforms of the material functions by
\begin{equation}  \label{eq:LT}
	\widetilde J(s) = \mathcal{L}\left[J(t); s\right], \qquad
	\widetilde G(s) = \mathcal{L}\left[G(t); s\right].
\end{equation}
For a causal linear viscoelastic material initially at rest, the creep compliance and the relaxation modulus provide equivalent representations of the constitutive law. In the one-dimensional case, under zero initial stress and strain, these may be written as \cite{MainardiBook}
\begin{equation}
	\epsilon(t) = \int_0^t J(t-\tau)\,\frac{d\sigma}{d\tau}(\tau)\,d\tau,
	\qquad
	\sigma(t) =	\int_0^t G(t-\tau)\,\frac{d\epsilon}{d\tau}(\tau)\,d\tau.
\end{equation}
Taking the Laplace transform gives
\begin{equation}
	\widetilde{\epsilon}(s)	= s\widetilde{J}(s)\widetilde{\sigma}(s),
	\qquad
	\widetilde{\sigma}(s) = s\widetilde{G}(s)\widetilde{\epsilon}(s),
\end{equation}
from which it follows that
\begin{equation}
	s^2 \widetilde{J}(s) \, \widetilde{G}(s) = 1,
	\qquad\text{or equivalently}\qquad
	s\widetilde{J}(s)=\frac{1}{s\widetilde{G}(s)}.
\end{equation}
By the convolution theorem, this is equivalent in the time domain to
\begin{equation}
	J(t)\ast G(t) = \int_0^t J(t-\tau)G(\tau)\,d\tau = 
	\int_0^t J(\tau)G(t-\tau)\,d\tau = t.
\end{equation}
The right-hand side follows from
\begin{equation}
	\mathcal{L}^{-1} \left[\frac{1}{s^2};t\right] = t.	
\end{equation}
We also introduce the instantaneous and equilibrium values
\begin{equation}
	J_g = J(0^+), \qquad J_e = J(\infty),
	\qquad
	G_g = G(0^+), \qquad G_e = G(\infty).
\end{equation}
When the instantaneous values are finite and nonzero, the identities
\begin{equation}
	J_g = \lim_{s\to\infty}s\widetilde J(s),
	\qquad
	G_g = \lim_{s\to\infty}s\widetilde G(s),
\end{equation}
together with the reciprocity relation, give
\begin{equation}
	J_g G_g = 1.
\end{equation}
Likewise, whenever the equilibrium values $J_e$ and $G_e$ are finite
and nonzero,
\begin{equation}
	J_e = \lim_{s\to0}s\widetilde J(s),
	\qquad
	G_e = \lim_{s\to0}s\widetilde G(s),
\end{equation}
and therefore
\begin{equation}
	J_e G_e = 1.
\end{equation}
Following Mainardi \cite{MainardiBook}, we classify viscoelastic bodies into four types according to their instantaneous and equilibrium responses (Table~\ref{t2.1}).
\begin{table}[htbp]
	\centering
	\begin{tabular}{|c||c|c||c|c|}
		\hline
		Type & $J_g$ & $J_e$ & $G_g$ & $G_e$ \\ \hline
		I & $>0$ & $<\infty$ & $<\infty$ & $>0$ \\
		II & $>0$ & $=\infty$ & $<\infty$ & $=0$ \\
		III & $=0$ & $<\infty$ & $=\infty$ & $>0$ \\
		IV & $=0$ & $=\infty$ & $=\infty$ & $=0$ \\ \hline
	\end{tabular}
	\caption{The four types of viscoelasticity.}
	\label{t2.1}
\end{table}
For mechanical models consisting of networks of springs and dashpots, the hereditary constitutive laws reduce to differential equations with constant coefficients. The simplest models of types I and II are the Zener and Maxwell models, respectively.

The Zener model is governed by the following stress--strain relation:
\begin{equation}
	\left[1 + \tau_\sigma\, \frac{d}{dt}\right] \sigma(t) =
	m \left[ 1 + \tau_\epsilon\, \frac{d}{dt}\right] \epsilon(t),
	\label{eq:Zener}
\end{equation}
Here $\tau_\epsilon$ is the strain retardation time under constant stress and $\tau_\sigma$ is the stress relaxation time under constant strain, with $\tau_\epsilon>\tau_\sigma>0$. The corresponding material functions are \cite[Eq.~2.19b]{MainardiBook}
\begin{equation}
	\begin{cases}
		{\displaystyle J(t) = J_g + J_1 \left( 1 - e^{-t/\tau_\epsilon}\right)}, &
		{\displaystyle J_g = \frac{1}{m} \frac{\tau_\sigma}{\tau_\epsilon}, \;
			J_1 = \frac{1}{m}\left(1- \frac{\tau_\sigma}{\tau_\epsilon}\right)}; \\
		{\displaystyle G(t) = G_e + G_1 \, e^{-t/\tau_\sigma}}, &
		{\displaystyle G_e = m, \; G_1 = m\left( \frac{\tau_\epsilon}{\tau_\sigma} -1\right)}.
	\end{cases}
	\label{eq:SLS JG  (2.19b)}
\end{equation}

The Maxwell model is governed by the following stress--strain relation:
\begin{equation}
	\sigma(t) + \tau_\sigma \, \frac{d\sigma}{dt} = \eta\, \frac{d\epsilon}{dt},
	\label{eq:Maxwell}
\end{equation}
Here $\tau_\sigma$ is the stress relaxation time under constant strain. The corresponding material functions are \cite[Eq.~2.18b]{MainardiBook}
\begin{equation} \label{eq:MAXWELL JG}
	\begin{cases}
		{\displaystyle J(t) = J_g + J_+\,t}, &
		{\displaystyle J_g= \frac{\tau_\sigma}{\eta}, \; J_+ = \frac{1}{\eta}}\,; \\
		{\displaystyle G(t) = G_1\, e^{-t/\tau_\sigma}}, &
		{\displaystyle G_1 = \frac{\eta}{\tau_\sigma}}.
	\end{cases}
\end{equation}

\section{Wave equations} \label{Section 3: Wave equations}

Plane waves in a perfectly elastic, homogeneous, one-dimensional medium satisfy the d'Alembert equation:
\begin{equation}  \label{eq:wave}
	\frac{\partial^2 r}{\partial t^2} = \frac{m}{\rho} \frac{\partial^2 r}{\partial x^2}.
\end{equation}
This hyperbolic equation has wave-front velocity $c=\sqrt{m/\rho}$. For the viscoelastic models considered here, a finite wave-front velocity requires $J_g>0$, as in types I and II, and is given by
\begin{equation}
	c = \sqrt{\frac{1}{J_g\,\rho}} = \sqrt{\frac{G_g}{\rho}},
	\label{c_def}
\end{equation}
The wave equation for the Zener model is a hyperbolic equation of third order in time and second order in space \cite[Eq.~4.28]{MainardiBook}:
\begin{equation}
	\left\{
	\begin{array}{l}
		{\displaystyle \frac{\partial}{\partial t}\left( \frac{\partial^2 r}{\partial t^2}
			- c^2\, \frac{\partial^2 r}{\partial x^2} \right) + \frac{1}{\tau_\sigma} \left( \frac{\partial^2 r}{\partial t^2}
			- c_0^2\, \frac{\partial^2 r}{\partial x^2} \right)} = 0, \\[6pt]
		c_0^2 = \dfrac{m}{\rho} = \dfrac{G_e}{\rho}, \quad
		c^2 = c_0^2 \, \dfrac{\tau_\epsilon}{\tau_\sigma} = \dfrac{G_g}{\rho},
	\end{array}
	\right.
	\label{eq:ZENER (4.28)}
\end{equation}
with characteristic velocity $c$ and subcharacteristic velocity $c_0$, as discussed by Chin \cite{Chin}.

For the Maxwell model, the wave equation reduces to the telegraph equation, which is second order in both time and space. It is a special case of the Klein--Gordon equation with dissipation studied in \cite{MainardiMentrelliSantander}:
\begin{equation}
	\frac{\partial^2 r}{\partial t^2} + \frac{1}{\tau_\sigma} \, \frac{\partial r}{\partial t}
	= c^2 \, \frac{\partial^2 r}{\partial x^2},
	\qquad c^2 = \frac{\eta}{\tau_\sigma\,\rho} = \dfrac{G_g}{\rho}.
	\label{eq:MAXWELL(4.27)}
\end{equation}

\subsection{The complex index of refraction}

We now consider wave propagation in a semi-infinite, homogeneous, linear viscoelastic medium initially at rest. Waves generated by a boundary impact are commonly termed \textit{impact waves}. The Laplace transform, defined by
\begin{equation}
	\mathcal{L}\left[f(t); s\right] = \int_{0}^{\infty} e^{-st} f(t)\,dt,
\end{equation}
provides an integral representation of the response in terms of a complex index of refraction.

Following \cite[Sect.~4.3.4]{MainardiBook} and \cite{MainardiMentrelliSantander}, we prescribe an impulsive boundary excitation at $x=0$:
\begin{equation} \label{r(0,t)=delta(t)}
	r(0,t) = \delta\!\left(t\right),
\end{equation}
where $\delta$ is the Dirac delta distribution. The function $r_\delta(x,t)$ denotes the impulse response of the field variable prescribed at the boundary. The analysis depends only on its propagation equation and Laplace-domain transfer function, so we retain the generic notation $r$.
The response is obtained by evaluating the Bromwich integral
\begin{equation}
		r_{\delta}(x,t) = \mathcal{L}^{-1}\left[
		\exp\left(-\frac{x}{c}\,s\,n(s)\right); t \right]
		\label{r_delta=L-1[n(s)]} 
		= \frac{1}{2\pi i} \int_{Br}
		\exp\left( s\left[ t - \frac{x}{c} n(s) \right] \right)\, ds,
\end{equation}
where $Br$ is the Bromwich path and $n(s)$ is the complex index of refraction. For the Zener model,
\begin{equation}
	n(s) = \sqrt{\frac{s + 1/\tau_{\sigma}}{s + 1/\tau_{\epsilon}}},
	\label{eq:SLS n(s) (4.51)}
\end{equation}
Here $\tau_\epsilon>\tau_\sigma>0$, and the square-root branch is chosen so that $n(s)\to1$ as $|s|\to\infty$. The Maxwell index follows in the limit $\tau_\epsilon\to\infty$ at fixed $\tau_\sigma$ and $c$:
\begin{equation} \label{eq:MAXWELL n(s)}
	n(s) = \sqrt{1 + \frac{1}{\tau_{\sigma}\,s}}, \qquad \tau_{\sigma} > 0.
\end{equation}

\section{Integral representation of the delta-pulse solution} \label{Section 4: Integral representation}

Equations~\eqref{r_delta=L-1[n(s)]} and \eqref{eq:SLS n(s) (4.51)} require the inverse Laplace transform
\begin{equation}
	r_{\delta }\left( x,t\right) =\mathcal{L}^{-1}\left[ \exp\!\left(
	-\chi s\sqrt{\frac{s+\alpha }{s+\beta }}\right);\, t\right],
\end{equation}
where we set
\begin{equation}
	\frac{1}{\tau _{\sigma}} = \alpha >\beta =\frac{1}{\tau_{\epsilon}} > 0, \qquad
	\chi =\frac{x}{c} > 0.  \label{parameters_conversion}
\end{equation}%
Using the identity
\begin{equation}
	-\chi s\sqrt{\frac{s+\alpha }{s+\beta }} 
	= -\chi \left( s+\beta \right) \sqrt{1+\frac{\alpha -\beta }{s+\beta }}
	+\chi \beta \sqrt{1+\frac{\alpha -\beta }{s+\beta}},
\end{equation}
and the shift property of the inverse Laplace transform,
\begin{equation}
	\mathcal{L}^{-1}\left[ F\left( s+a\right);\, t\right] =e^{-at}\mathcal{L}^{-1}
	\left[ F\left(s\right);\, t\right],  \label{Translation_property}
\end{equation}
we write
\begin{equation}
	r_{\delta }\left( x,t\right)  \label{r_delta_L-1} 
	= e^{-\beta t}\,\mathcal{L}^{-1}\left[ \exp\!\left(-\chi \sqrt{s^2 + \left(\alpha -\beta\right) s}\right) \exp\!\left(\chi \beta \sqrt{1 + \frac{\alpha -\beta}{s}}\right);\, t\right].
\end{equation}
After defining
\begin{equation} \label{u(a,b,t)_def}
	u_{a, b}(t) = \mathcal{L}^{-1}\left[ \exp\!\left(-a\sqrt{s^2 + bs}\right);\, t\right],
\end{equation}
\begin{equation} \label{v(a,b,t)_def}
	v_{a, b}(t) =\mathcal{L}^{-1}\left[ \exp\!\left(-a\sqrt{1 + \frac{b}{s}}\right);\, t\right],
\end{equation}
by means of the convolution theorem
\begin{equation}
	\mathcal{L}^{-1}\left[ \mathcal{L}\left[ f\left( t\right) ;s\right]\,%
	\mathcal{L}\left[ g\left( t\right) ;s\right];\, t\right] =\int_{0}^{t}f\left(\tau\right) \,g\left(t - \tau\right) \,d\tau,
\end{equation}
we obtain
\begin{equation} \label{r(x,t)_1}
	r_{\delta}\left(x,t\right) = 
	e^{-\beta t} \int_{0}^{t} u_{\chi, \alpha - \beta}(\tau)\, 
	v_{-\chi \beta, \alpha -\beta}(t-\tau)\, d\tau.
\end{equation}

We next compute the inverse transforms in \eqref{u(a,b,t)_def} and \eqref{v(a,b,t)_def}.

\subsection{Explicit expression of $u_{a, b}(t)$}

For $a,b>0$, set $\xi=b/2$ and use the shift property to write
\begin{equation}
	\begin{split}
		u_{a, b}(t) = 
		\mathcal{L}^{-1}\left[ \exp\!\left(-a\sqrt{\left(s + \xi\right)^2 - \xi^2}\right);\, t\right] 
		= \exp\!\left(-\xi t\right) \,\mathcal{L}^{-1}\left[ \exp\!\left(-a\sqrt{s^2 - \xi^2}\right);\, t\right].
	\end{split}
\end{equation}
According to \cite[Eq.~2.2.5(8)]{Prudnikov5} (see also \cite{MainardiMentrelliSantander}), for $a>0$,
\begin{equation}
	\mathcal{L}^{-1}\left[ \frac{\exp\!\left(-a \sqrt{s^2-\xi^2}\right) }{\sqrt{s^2 - \xi^2}};\, t\right] = I_0\!\left( \xi \sqrt{%
		t^2 - a^2}\right)\, \theta\!\left(t - a\right).
\end{equation}
Here $I_\nu$ denotes the modified Bessel function of the first kind of order $\nu$, and $\theta$ is the Heaviside function. Differentiating this identity with respect to $a$ in the sense of distributions gives
\begin{equation} \label{resultado_1}
	\begin{split}
		\mathcal{L}^{-1}\left[ \exp\!\left( -a\sqrt{s^2 - \xi^2}\right);\, t\right]  
		&= -\frac{d}{da}\left[ I_0\!\left(\xi\sqrt{t^2 - a^2}\right)\, \theta\!\left(t - a\right) \right] \\
		 &= a \xi \frac{I_1\!\left(\xi\sqrt{t^2 - a^2}\right)}{\sqrt{t^2 - a^2}}\, \theta\!\left(t - a\right) + I_0\!\left(\xi\sqrt{t^2 - a^2}\right)\, \delta\!\left(t - a\right),
	\end{split}
\end{equation}
and hence
\begin{equation} \label{u(a,b,t)_resultado}
	u_{a, b}(t) = 
	\exp\!\left( -\frac{bt}{2}\right) 
	\left[ ab\,\frac{I_1\!\left(\frac{b}{2}\sqrt{t^2 - a^2}\right) }{2\sqrt{t^2 - a^2}} \theta\!\left(t-a\right) + I_0\!\left( \frac{b}{2}\sqrt{t^2 - a^2}\right) \delta\!\left( t - a\right)\right].
\end{equation}

\subsection{Explicit expression of $v_{a, b}(t)$}

\begin{lemma} \label{Lemma1}
	For every $b>0$, the following inverse Laplace-transform formula holds:
	\begin{equation} \label{eq:lemma_inverse_F_b_over_s}
		\mathcal{L}^{-1}\!\left[F\!\left(\frac{b}{s}\right);t\right] = 
		C\delta\!\left(t\right) - \theta\!\left(t\right) \int_0^\infty \sqrt{\frac{bx}{t}}\, J_1\!\left(2\sqrt{btx}\right)f(x)\,dx,
	\end{equation}
	where $J_1$ is the Bessel function of the first kind of order one, $\theta$ is the Heaviside function, $f(x)=\mathcal{L}^{-1}[F(s);x]$,
	\begin{equation} \label{C_def}
		C = \int_0^\infty f(x)\,dx = F(0^+),
	\end{equation}
	provided that $f \in L_{\textrm{loc}}^1\left(\left[0, \infty\right)\right)$
	satisfies
	\begin{equation} \label{eq:M12}
		\int_0^\infty \sqrt{x}\,|f(x)|\,dx<\infty.
	\end{equation}
	The equality holds in the sense of distributions, with Laplace	transforms converging for $\realpart{s} > 0$.
\end{lemma}
\begin{proof}
	For $t,x>0$, set
	\begin{equation}
		K_b(t,x) = \sqrt{\frac{bx}{t}}\, J_1\!\left(2\sqrt{btx}\right), 
		\qquad
		g_b(t) = \theta\!\left(t\right)\int_0^\infty K_b(t,x)f(x)\,dx.
	\end{equation}
	Since $|J_1(y)|\leq 1$ for $y\geq 0$, the integral defining $g_b(t)$ converges absolutely for every $t>0$, and
	\begin{equation}
		|g_b(t)|\leq\sqrt{\frac{b}{t}}\,M_{1/2}, \qquad
		M_{1/2}:=\int_0^\infty\sqrt{x}\,|f(x)|\,dx<\infty.
	\end{equation}
	Thus $g_b$ is locally integrable. Moreover, for every $\xi>0$,
	\begin{equation}
		\int_0^\infty\!\int_0^\infty
		e^{-\xi t}|K_b(t,x)f(x)|\,dx\,dt
		\leq \sqrt{b} M_{1/2} \int_0^\infty e^{-\xi t} t^{-1/2} dt =
		\sqrt{\frac{\pi b}{\xi}}\,M_{1/2} < \infty.
	\end{equation}
	Fubini's theorem therefore permits the interchange of the $x$- and $t$-integrations when $\realpart{s}>0$. In particular,
	\begin{equation} \label{eq:L(s)}
		L(s):=\mathcal{L}\left[g_b(t); s\right]
		=\int_0^\infty f(x) \left\{\int_0^\infty e^{-st}K_b(t,x)\,dt\right\}dx.
	\end{equation}
	For each $x>0$, the power series of $J_1$ gives
	\begin{equation} \label{eq:Kb}
		K_b(t,x) = \sum_{n=0}^\infty \frac{(-1)^n(bx)^{n+1}t^n}{n!(n+1)!},
	\end{equation}
	and term-by-term Laplace transformation of \eqref{eq:Kb} yields
	\begin{equation} \label{eq:L-Kb}
		\int_0^\infty e^{-st}K_b(t,x)\,dt 
		= \sum_{n=0}^\infty \frac{(-1)^n(bx/s)^{n+1}}{(n+1)!} 
		= 1-e^{-bx/s},
	\end{equation}
	Term-by-term integration is justified because, for $\xi=\realpart{s}>0$, the sum of the integrals of the absolute terms is $e^{bx/\xi}-1<\infty$.
	
	Substituting \eqref{eq:L-Kb} into \eqref{eq:L(s)} gives
	\begin{equation} \label{L(s)_resultado_1}
		L(s) =\int_0^\infty f(x)\,dx - \int_0^\infty e^{-bx/s}f(x)\,dx
		= C - F\!\left(\frac{b}{s}\right), 
	\end{equation}
	The value $F(b/s)$ is well defined because $\realpart{b/s} > 0$.
	
	Finally, recalling that $\mathcal{L}[\delta\!\left(t\right); s] = 1$, equations \eqref{eq:L(s)} and \eqref{L(s)_resultado_1} give
	\begin{equation} \label{eq:F(b/s)}
		\mathcal{L}[C\delta\!\left(t\right) - g_b(t); s] = F\!\left(\frac{b}{s}\right),
	\end{equation}
	which proves the assertion.
\end{proof}

\begin{theorem} \label{Theorem}	
	For $a, b > 0$, the following inverse Laplace transform holds:
	\begin{equation} \label{eq:theorem}
		\begin{split}
			v_{a,b}(t) &=\mathcal{L}^{-1}\left[\exp\!\left(-a\sqrt{1+\frac{b}{s}}\right); t\right]\\
			&=e^{-a}\delta\!\left(t\right) - \theta\!\left(t\right)\frac{a^{3/2}b}{\sqrt{2\pi}} \sum_{n=0}^{\infty} \frac{(-abt/2)^nK_{n+1/2}(a)}{n!(n+1)!},
		\end{split}
	\end{equation}
	The square root is taken on its principal branch, and $K_\nu$ denotes the Macdonald function (modified Bessel function of the second kind) of order $\nu$. The series converges absolutely for every $t\geq0$ and uniformly on every bounded interval in $[0,\infty)$.
\end{theorem}
\begin{proof}
	For $a>0$, define
	\begin{equation}
		F_a(z)=e^{-a\sqrt{1+z}}, \qquad \realpart{z} > 0.
	\end{equation}
	The standard inverse-transform formula for $e^{-a\sqrt{z}}$ \cite[Eq.~2.2.1(9)]{Prudnikov5}, followed by the shift property, gives
	\begin{equation} \label{L-1->exp(1+s)}
		f_a(x):=\mathcal{L}^{-1}[F_a(z);x] = 
		\frac{a}{2\sqrt{\pi}}x^{-3/2} \exp\!\left(-x-\frac{a^2}{4x}\right), 
		\qquad x>0.
	\end{equation}
	Since $a>0$, the function $f_a$ extends continuously to $x=0$ with $f_a(0)=0$, and therefore belongs to $L^1_{\mathrm{loc}}([0,\infty))$. Its exponential decay at infinity also ensures that
	\begin{equation}
		\int_0^\infty \sqrt{x}\, |f_a(x)|\,dx<\infty,
	\end{equation}
	so $f_a$ satisfies the hypotheses of Lemma~\ref{Lemma1}. Moreover,
	\begin{equation} \label{L-1->exp(-a)}
		C_a = \int_0^\infty f_a(x)\,dx = F_a(0^+) = e^{-a},
	\end{equation}
	and Lemma~\ref{Lemma1} therefore yields
	\begin{equation}
		v_{a,b}(t) = e^{-a}\delta\!\left(t\right)+\theta\!\left(t\right) w_{a,b}(t),
	\end{equation}
	where, for $t>0$,
	\begin{equation} \label{L-1->exp(1+b/s)}
		w_{a,b}(t) = -\frac{a}{2}\sqrt{\frac{b}{\pi t}} 
		\int_0^\infty J_1\!\left(2\sqrt{btx}\right) \exp\!\left(-x-\frac{a^2}{4x}\right)\frac{dx}{x}.
	\end{equation}
	
	To expand this integral, use the power series \cite[Eq.~5.3.2]{Lebedev}
	\begin{equation}
		J_1\!\left(2\sqrt{btx}\right) = \sqrt{btx}\sum_{n=0}^{\infty} \frac{(-btx)^n}{n!(n+1)!}.
	\end{equation}
	The integral representation of $K_\nu$ \cite[Eq.~5.10.25]{Lebedev}, with $\nu=-n-\tfrac12$, and the symmetry $K_{-\nu}=K_\nu$ \cite[Eq.~5.7.10]{Lebedev} give
	\begin{equation} \label{eq:K_12}
		K_{n+1/2}(a) = \frac{1}{2} \left(\frac{2}{a}\right)^{n+1/2}
		\int_0^\infty x^{n-1/2}
		\exp\!\left(-x-\frac{a^2}{4x}\right)\,dx,
	\end{equation}
	which gives, for every $n\geq 0$,
	\begin{equation} \label{eq:K_12_bound}
		2\left(\frac{a}{2}\right)^{n+1/2}K_{n+1/2}(a)
		\leq  \int_0^\infty x^{n-1/2} e^{-x} \,dx = \Gamma\!\left(n+\frac12\right)
		\leq \sqrt{\pi}\,n!,
	\end{equation}
	where the last inequality follows from $\Gamma(1/2)=\sqrt{\pi}$ and
	\begin{equation}
		\frac{\Gamma(n+\tfrac12)}{n!}
		=\sqrt{\pi}\prod_{k=1}^{n}\frac{k-1/2}{k}
		\leq\sqrt{\pi},
	\end{equation}
	with the empty product equal to $1$ when $n=0$.
	
	For every $T>0$, \eqref{eq:K_12_bound} gives
	\begin{equation}
		\frac{a^{3/2}b}{\sqrt{2\pi}} 
		\sum_{n=0}^{\infty} \frac{(abT/2)^nK_{n+1/2}(a)}{n!(n+1)!}
		\leq \frac{ab}{2}\sum_{n=0}^{\infty} \frac{(bT)^n}{(n+1)!} < \infty,
	\end{equation}
	which bounds the sum of the integrals of the absolute terms in the Bessel expansion of \eqref{L-1->exp(1+b/s)}, uniformly for $0<t\leq T$. Term-by-term integration therefore yields \eqref{eq:theorem}. The same estimate proves absolute and uniform convergence on $[0,T]$, and the regular series extends $w_{a,b}$ continuously to $t=0$.
\end{proof}

\begin{remark} \label{Remark1}
	Theorem~\ref{Theorem} extends to every real $a$, including $a \leq 0$, by analytic continuation in $a$. To specify this extension, introduce the polynomials
	\begin{equation}
		P_n(a) = \sum_{k=0}^{n} \frac{(n+k)!}{2^k k!(n-k)!}\,a^{n-k}.
	\end{equation}
	The half-integer Macdonald formula gives, for $a>0$,
	\begin{equation}
		a^{n+3/2}K_{n+1/2}(a) = \sqrt{\frac{\pi}{2}}\,a e^{-a}P_n(a),
	\end{equation}
	and \eqref{eq:theorem} can be rewritten as
	\begin{equation}
		v_{a,b}(t)=e^{-a}\delta\!\left(t\right)+\theta\!\left(t\right)w_{a,b}(t),
	\end{equation}
	with
	\begin{equation}
		w_{a,b}(t)=-\frac{ab}{2} e^{-a} \sum_{n=0}^{\infty} \frac{(-bt/2)^nP_n(a)}{n!\,(n+1)!},
		\qquad t\geq 0.
	\end{equation}
	Every coefficient in this series is entire in $a$.
	
	For $A>0$ and complex $a$ with $|a|\leq A$, positivity of the polynomial coefficients of $P_n(a)$ and the estimate \eqref{eq:K_12_bound} give
	\begin{equation}
		|P_n(a)|\le P_n(A)\le 2^n e^A n!.
	\end{equation}
	Thus, the absolute value of the $n$-th summand, including the prefactor $-ab e^{-a}/2$, is bounded by
	\begin{equation}
		\frac{Ab e^{2A}}2 \frac{(bt)^n}{(n+1)!}. 
	\end{equation}
	The continued series therefore converges locally uniformly in $a$ and uniformly for $t$ in bounded intervals. Its sum $w_{a,b}(t)$ is entire in $a$ and satisfies
	\begin{equation}
		|w_{a,b}(t)|\le \frac{Ab e^{2A}}2 e^{bt}, \qquad |a| \leq A, \quad t\geq 0.
	\end{equation}
	
	For each $s$ with $\realpart{s} > b$, both
	\begin{equation}
		e^{-a}+\int_0^\infty e^{-st}w_{a,b}(t)\,dt
		\qquad\text{and}\qquad
		\exp\!\left(-a\sqrt{1+\frac bs}\right)
	\end{equation}
	are therefore entire in $a$. They agree for $a>0$ by Theorem~\ref{Theorem}, so the identity theorem extends their equality to every $a\in\mathbb C$. Hence \eqref{eq:theorem}, interpreted through the entire coefficients above, holds for every real $a$. 
\end{remark}

\begin{remark} \label{Remark2}
	The Macdonald function of half-integer order can be expressed in terms of elementary functions as follows \cite[Eq.~8.440]{Gradshteyn}:
	\begin{equation}
		K_{n+1/2}\left(z\right) =
		\sqrt{\frac{\pi}{2z}} e^{-z} \sum_{k=0}^{n}
		\frac{\left(n+k\right)!}{k!\, \left(n-k\right)!\, \left(2z\right)^k}.
	\end{equation}
	%
	%
	For $a>0$, the regular series in \eqref{eq:theorem} is alternating. For the negative values $a=-\chi\beta$ needed below, the polynomial representation in Remark~\ref{Remark1} applies. It converges absolutely and uniformly on bounded time intervals, although its terms need not alternate in sign.
	
	For numerical evaluation, the estimate in Remark~\ref{Remark1} gives an absolute truncation-error bound. If $w_{a,b}^{(N)}(t)$ denotes the sum through $n=N$, including the prefactor, then, for $A>0$ and $t\geq0$,
	\begin{equation}
		| w_{a,b}(t) - w_{a,b}^{(N)}(t) | \leq \frac{A b e^{2A}}{2} \sum_{n=N+1}^{\infty} \frac{\left(bt\right)^n}{\left(n+1\right)!}, \qquad |a| \leq A.
	\end{equation}
\end{remark}

\subsection{Explicit expression of the impulse response $r_\delta(x,t)$}

Using \eqref{u(a,b,t)_resultado}, \eqref{eq:theorem}, and Remark~\ref{Remark1}, write
\begin{eqnarray}
	u_{a, b}(t) &=& u_{a, b}^{I}(t)\,\delta\!\left(t - a\right) + 
	u_{a, b}^{I\!I}(t)\,\theta\!\left(t - a\right),  \label{u_split} \\
	v_{a, b}(t) &=& e^{-a}\delta\!\left(t\right) + w_{a,b}(t)\,\theta\!\left(t\right),  \label{v_split}
\end{eqnarray}%
where
\begin{eqnarray}
	u_{a, b}^{I}(t) &=& 
		I_0\!\left(\frac{b}{2}\sqrt{t^2 - a^2}\right) \exp\!\left(-\frac{bt}{2}\right), \label{u2_def} \\
	u_{a, b}^{I\!I}(t) &=&
		ab\,\frac{I_1\!\left(\frac{b}{2}\sqrt{t^2 - a^2}\right)}{2\sqrt{t^2 - a^2}}\exp\!\left(-\frac{bt}{2}\right), \label{u1_def} \\
	w_{a,b}(t) &=& 
		-\frac{ab}{2}e^{-a}\sum_{n=0}^{\infty}\frac{(-bt/2)^nP_n(a)}{n!\,(n+1)!}. \label{v1_def}
\end{eqnarray}

Here $w_{a,b}(t)$ is defined for every real $a$ and $t\geq0$, whereas the formulas for $u_{a,b}^{I}$ and $u_{a,b}^{I\!I}$ are used for $a>0$. The apparent singularity of $u_{a, b}^{I\!I}(t)$ at $t=a$
is removable. Since $I_1(z)/z\to1/2$ as $z\to0$,
\begin{equation}
    \lim_{t\to a^+} u_{a, b}^{I\!I}(t) = \frac{ab^2}{8}e^{-ab/2}.
\end{equation}
We use this value to define $u_{a, b}^{I\!I}(t)$ by continuity.

Substituting \eqref{u_split}--\eqref{v_split} into \eqref{r(x,t)_1} and recalling that $I_0\!\left(0\right) = 1$ gives
\begin{equation} \label{eq:integral_representation}
	\begin{split}
		r_{\delta}(x,t) = 
		&e^{-\chi\left(\alpha - \beta\right)/2}\,
		\delta\!\left(t-\chi\right) + \\
		&+ \theta\!\left(t - \chi\right)\, e^{-\beta t} \Biggl\{
		\int_{\chi}^{t} u_{\chi, \alpha-\beta}^{I\!I}(\tau)\,
		w_{-\chi\beta, \alpha-\beta}(t-\tau)\, d\tau + \phantom{\Biggr\}} \\
		&\qquad\qquad\qquad \phantom{\Biggl\{} + e^{-\chi\left(\alpha - \beta\right)/2}\,
		w_{-\chi\beta, \alpha-\beta}(t-\chi) +
		e^{\chi\beta}\, u_{\chi, \alpha-\beta}^{I\!I}(t) \Biggr\}
	\end{split}
\end{equation}
Separating the singular wave-front contribution from the regular part, we rewrite \eqref{eq:integral_representation} as
\begin{equation} \label{eq:rdelta}
	r_{\delta}(x,t) = A_{\mathrm f}(x)\,\delta\!\left(t-\chi\right) + r_{\delta,\mathrm{reg}}(x,t) \, \theta\!\left(t-\chi\right),
\end{equation}
where
\begin{equation} \label{eq:Af}
	A_{\mathrm f}(x) = \exp\!\left[-\frac{\alpha-\beta}{2}\chi\right]
	= \exp\!\left[-\frac{x}{2c}
	\left(\frac{1}{\tau_{\sigma}} - \frac{1}{\tau_{\epsilon}}
	\right)\right],
\end{equation}
and, for $t\geq\chi$,
\begin{equation} \label{eq:r_delta_reg}
	\begin{split}
		r_{\delta,\mathrm{reg}}(x,t) = 
		&e^{-\beta t} \Biggl\{
		\int_{\chi}^{t} u_{\chi, \alpha-\beta}^{I\!I}(\tau)\,
		w_{-\chi\beta, \alpha-\beta}(t-\tau)\, d\tau + \phantom{\Biggr\}} \\
		&\qquad\qquad \phantom{\Biggl\{} + e^{-\chi\left(\alpha - \beta\right)/2}\,
		w_{-\chi\beta, \alpha-\beta}(t-\chi) +
		e^{\chi\beta}\, u_{\chi, \alpha-\beta}^{I\!I}(t) \Biggr\}.
	\end{split}
\end{equation}

To recover the Maxwell model, let $\beta\to0$ at fixed $\alpha$ and $c$. At fixed density, this is equivalent to taking $\tau_\epsilon\to\infty$ with $\tau_\sigma$ and $\eta=m\tau_\epsilon$ fixed, so that $m\to0$ and $G_g=\eta/\tau_\sigma=\rho c^2$ remains fixed. The polynomial representation and bound in Remark~\ref{Remark1} imply
\begin{equation}
    \lim_{\beta\to0} w_{-\chi\beta, \alpha-\beta}(q) = 0
\end{equation}
uniformly for $q$ in bounded subsets of $[0,\infty)$. Thus the terms involving $v^I$ in \eqref{eq:integral_representation}, including the convolution integral, vanish in the limit, and
\begin{equation}
	\lim_{\beta \rightarrow 0} r_{\delta }\left( x,t\right) = 
	e^{-\alpha \chi/2}\, \delta\!\left(t - \chi \right) + u_{\chi, \alpha}^{I\!I}(t)\, \theta\!\left(t - \chi \right).
\end{equation}%
Using \eqref{u1_def}, we recover the known Maxwell response \cite[Eq.~5.71]{MainardiBook}:
\begin{equation}
	\lim_{\beta \rightarrow 0}r_{\delta }\left( x,t\right) 
	=\exp\!\left( -\frac{\alpha t}{2}\right) \left[ \delta\!\left( t-\chi \right) + \chi \alpha \, \frac{I_1\!\left(\frac{\alpha}{2}\sqrt{t^2 - \chi^2}\right)}{2\sqrt{t^2 - \chi^2}}\theta\!\left(t - \chi \right)\right].
\end{equation}

\section{Delta-pulse solution via steepest-descent method} \label{Section 5: SDP}

We now evaluate the Bromwich integral \eqref{r_delta=L-1[n(s)]} by deforming the integration contour onto a \textit{steepest-descent path} (SDP) in the complex plane.

\subsection{The steepest-descent method: general framework}

The steepest-descent method (also called the saddle-point method) \cite{Brillouin, MainardiSPD} can be applied to contour integrals of the form
\begin{equation} \label{eq:Bromwich_general}
	r_\delta\!\left(x,t\right) = \frac{1}{2\pi i} \int_{Br} \exp\!\left(s\!\left[t -
	\frac{x}{c}\,n(s)\right]\right) ds,
\end{equation}
where $Br$ denotes the Bromwich path (a vertical line lying to the right of all singularities of the integrand), and $n(s)$ is the complex index of refraction of the medium. The idea is to deform the Bromwich path into a new contour $\gamma_\mu$ (the steepest-descent path, SDP), passing through a saddle point of the phase function
\begin{equation}
	F_\mu(s) = s \left[1 - \mu \, n(s)\right], \qquad
	\mu = \frac{x}{ct}, \qquad
	\left(0 \leq \mu \leq 1\right),
\end{equation}
i.e., a point $p$ in the complex plane satisfying $dF_\mu/ds = 0$. The path $\gamma_\mu$ is chosen so that:
\begin{itemize}
	\item it passes through a saddle point $p$;
	\item the imaginary part of $F_\mu$ is constant along $\gamma_\mu$;
	\item the real part of $F_\mu$ attains its maximum on $\gamma_\mu$ at the saddle point $p$;
	\item the deformed integral reproduces the Bromwich integral once any additional singular contributions required by the deformation are included.
\end{itemize}
The validity of such a procedure requires that the contributions from the closing parts of the contour vanish and that all singularities enclosed or crossed during the deformation be properly accounted for. These conditions are verified explicitly for the Zener model in Section~\ref{sec:5.2}.

\smallskip 

The contour-deformation principle itself is not restricted to the Zener model, nor to contour integrals of the form \eqref{eq:Bromwich_general}. More generally, a Bromwich integral may be written in the form
\begin{equation}
	f(t) = \frac{1}{2\pi i} \int_{\mathrm{Br}} A(s)\exp\left[\Phi(s;t)\right]\,ds,
\end{equation}
where $A(s)$ is an amplitude factor and $\Phi(s;t)$ is the phase function. In the present problem, $A(s)=1$ and
\begin{equation}
	\Phi(s;t) = s\left[t-\frac{x}{c}n(s)\right] = t F_\mu(s).	
\end{equation}
When $A(s)$ is nonconstant, the contour can still be chosen from the phase $\Phi$, with $A(s)$ retained in the quadrature. The integrand must admit analytic continuation in the deformation region, the closing contributions must vanish, and any singularities crossed or enclosed must be accounted for. Poles and branch points of $A(s)$ may require additional residue or branch-cut contributions. The amplitude need not vary slowly, although rapid variation or oscillation may reduce the numerical advantage of choosing the contour solely from $\Phi$. Several saddles or branch cuts may require multiple contour components. Thus, the deformation principle is general, while an explicit analytical description of the SDP depends on the integrand.

\smallskip 

For \eqref{eq:Bromwich_general}, the constant imaginary part of $F_\mu$ along $\gamma_\mu$ removes the rapid oscillations of the exponential factor and facilitates numerical quadrature.

In the Zener problem, the relevant contour $\gamma_\mu$ encloses the finite branch cut of $n(s)$. Integrating the complete contour captures the full regular response rather than only its local saddle-point approximation.

The contour geometry described below is specific to the Zener model. For $0<\mu<1$, the relevant closed constant-phase curve passes through two real saddle points and encloses the branch cut. As shown in Section~\ref{sec:5.2}, it lies in the level set
\begin{equation}
	\imagpart{F_\mu(s)} = \imagpart{s - s \, \mu \, n(s)} = 0.
\end{equation}
For more general viscoelastic models, the number and location of the relevant saddle points, branch cuts, and contour components may be different, and the steepest-descent path will generally have to be determined numerically.

Once $\gamma_\mu$ has been determined, analytically or numerically, the inverse transform reduces to a real-line integral evaluated by standard quadrature. For the Zener model, this gives an exact representation of the regular response for $x>0$ and $t>x/c$, without requiring large $x$ or $t$. A local saddle-point expansion can subsequently be used to obtain asymptotic approximations. Numerical accuracy near the limiting configurations $\mu=0$ and $\mu=1$ requires separate consideration.

The SDP approach does not require a closed-form inverse Laplace transform. It is therefore useful for models in which no integral representation in terms of known special functions is available.
A practical numerical construction of the SDP in such cases is
outlined in Section~\ref{Section 6: Conclusions}.

\subsection{Application of the steepest-descent method to the Zener model equations \label{sec:5.2}}

For the Zener index \eqref{eq:SLS n(s) (4.51)}, the phase per unit time is
\begin{equation} \label{eq:Fmu}
	F_{\mu}\left(s\right) = s\left(1 - \mu \sqrt{\frac{s + 1/\tau_{\sigma}}{s + 1/\tau_{\epsilon}}}\right),
	\qquad \mu = \frac{\chi}{t} = \frac{x}{ct}, \qquad \left( 0 \leq \mu \leq 1\right),
\end{equation}%
We construct the SDP for $0<\mu<1$, treating the endpoint configurations separately. The saddle points satisfy
\begin{equation} \label{eq:dFds}
	\frac{dF_{\mu }(s)}{ds} = 0,
\end{equation}
which, using \eqref{eq:Fmu}, becomes
\begin{equation} \label{eq:intermediate}
	\mu \left( \frac{s + 1/\tau_{\sigma}}{s + 1/\tau_{\epsilon}}\right)^{1/2} \left(1 - \frac{1}{2\left(s\tau_{\sigma} + 1\right)} + \frac{1}{2\left(s\tau_{\epsilon} + 1\right)}\right) - 1 = 0.
\end{equation}%

\begin{figure}
	\begin{center}
		\includegraphics[scale=0.4]{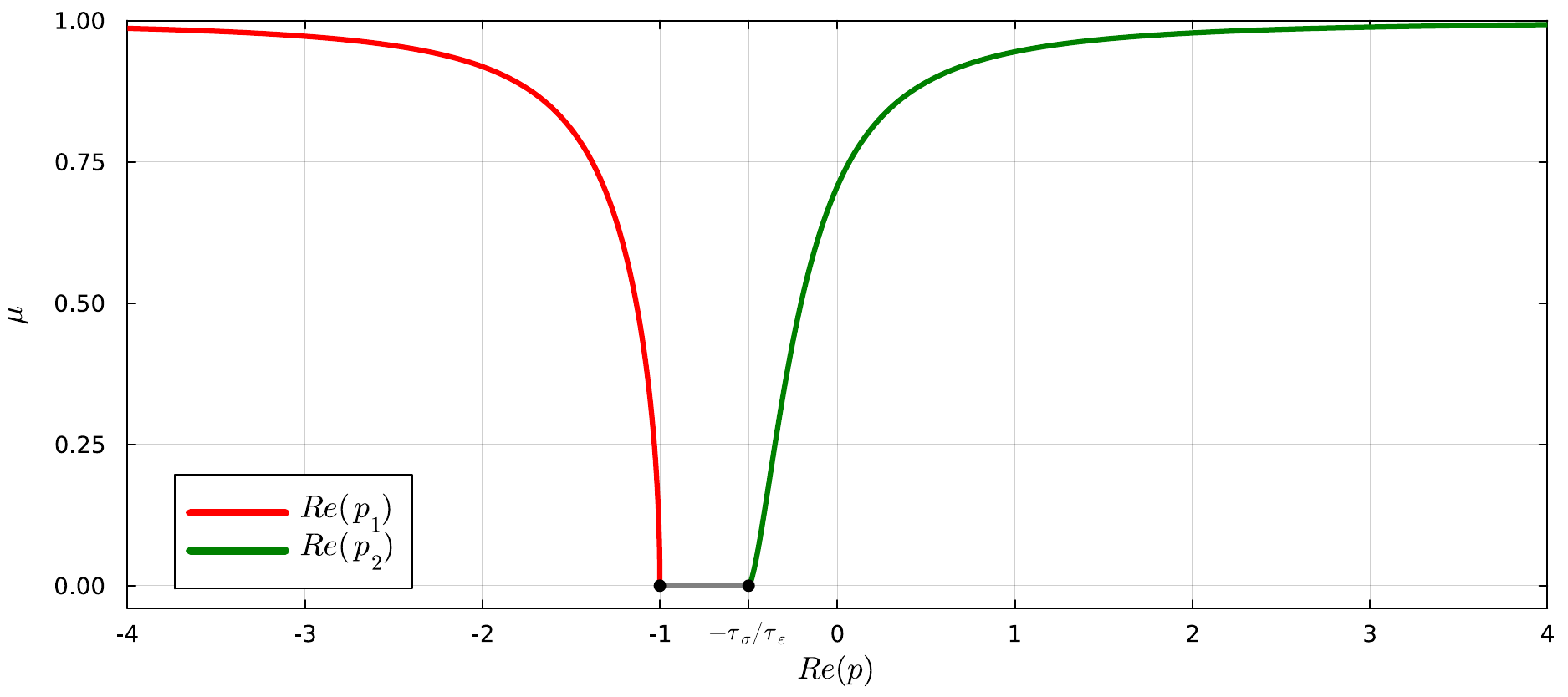}	
	\end{center}
	\caption{Real part of the saddle points $p_{1}$ and $p_{2}$ as functions of $\mu$ ($\text{Im}\left(p_1\right) = \text{Im}\left(p_2\right) = 0$). The black dots represent the branch points; the gray line represents the branch cut ($\tau_{\sigma} = 1$, $\tau_{\epsilon} = 2$). \label{Figure: saddle points p1 p2}}
\end{figure}

\begin{figure}
	\begin{center}
		\includegraphics[scale=0.4]{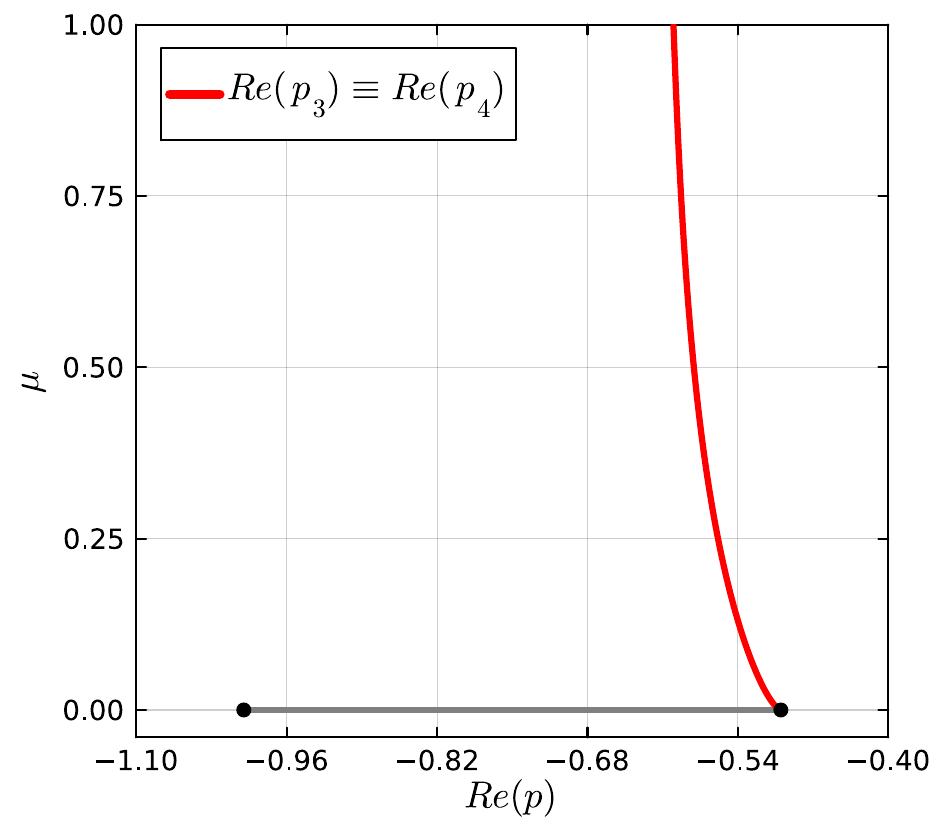}
		\includegraphics[scale=0.4]{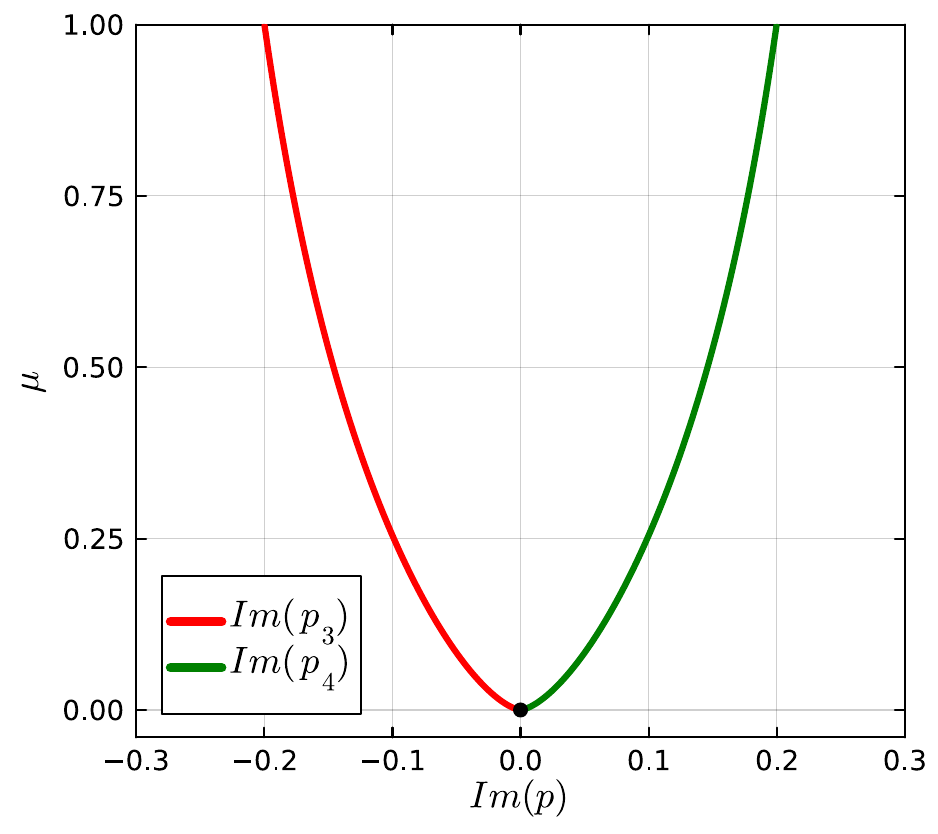}
	\end{center}
	\caption{Real part (left) and imaginary part (right) of the complex-conjugate saddle points $p_3$ and $p_4$ belonging to the second Riemann sheet of the analytically continued phase function, as functions of $\mu$. These points do not belong to the steepest-descent contour used for the Bromwich inversion on the physical sheet. The black dots represent the branch points and the gray line represents the branch cut ($\tau_\sigma=1$, $\tau_\epsilon=2$).\label{Figure: saddle points p3 p4}}
\end{figure}

Equation~\eqref{eq:intermediate} uses the square-root branch for which $n(s)\to1$ as $|s|\to\infty$. For $0<\mu<1$, squaring and clearing denominators gives a quartic equation. Two roots, $p_1$ and $p_2$, are real and satisfy the original equation on the physical sheet. The other two, $p_3$ and $p_4$, form a complex-conjugate pair and are stationary points on the second Riemann sheet.

The points $p_3$ and $p_4$ do not lie on the physical sheet used for the Bromwich inversion and give no separate contribution to \eqref{eq:Bromwich_general}. They are included in Figs.~\ref{Figure: saddle points p3 p4}--\ref{fig:case3} to illustrate the complete saddle structure of the two-sheeted continuation. Their imaginary parts tend to zero as $\mu\to0$.

The real saddles lie outside the branch cut, with $p_1<-\alpha$ and $p_2>-\beta$. Consequently,
\begin{equation}
	\imagpart{F_\mu\left(p_1\right)} = \imagpart{F_\mu\left(p_2\right)} = 0.
\end{equation}
We select the closed constant-phase curve through $p_1$ and $p_2$ that encloses the branch cut. It lies in the zero-level set
\begin{equation} \label{eq:ImagF=0}
	\imagpart{s - s \mu \sqrt{\frac{s + 1/\tau_{\sigma}}{s + 1/\tau_{\epsilon}}} } = 0.
\end{equation}%

Writing $s=\xi+i\eta$, we describe this family of contours, parametrized by $\mu$, through the implicit equation
\begin{equation}
	\gamma_{\mu}\left( \xi ,\eta \right) = 0.  \label{eq:gammamu}
\end{equation}%

Setting $\Xi = \tau_{\sigma}\xi$ and $\varphi = \tau_\sigma / \tau_\epsilon$, it can be shown that
\begin{equation}
	\gamma_{\mu}\left( \xi ,\eta \right) = C_{6}\, \eta^{6} + C_{4}\, \eta^{4} + C_{2}\, \eta^2 + C_{0},  \label{eq:poly}
\end{equation}%
where $C_{k}\left( \xi ,\mu ,\varphi ,\tau_{\sigma}\right) = \frac{1}{4}\tau_{\sigma}^{k-6}c_{k}\left( \Xi ,\mu ,\varphi \right)$ and
\begin{eqnarray*}
	c_{0} &=& \Xi^2\mu^2\left\{ 4\varphi^2\left( \varphi - \mu^2\right) - 4\Xi^{3}\left( \mu^2-1\right)\left( \Xi + 3\varphi + 1\right) \right. \\
	&&+ \, 4\Xi \varphi \left[ \varphi \left( \varphi + 3\right) - \mu^2\left( 3\varphi + 1\right) \right] \\
	&&-\left. \Xi^2\left[ \mu^2\left( 1 + \varphi \left( 9\varphi + 14\right) \right) - 12\varphi \left( 1 + \varphi \right) \right] \right\}, \\
	c_{2} &=& 4\Xi^{4}\left( 1 + \mu^2 - 2\mu^{4}\right) + 4\varphi^{3}\left( \varphi - \mu^2\right) - 8\Xi^{3}\left( \mu^2-1\right)\left[ 2\varphi + \mu^2\left( \varphi + 1\right) \right] \\
	&&+ \, 4\Xi \varphi \left[ \mu^{4}\left( \varphi - 1\right) + 4\varphi^2 - \mu^2\varphi \left( 3\varphi + 1\right) \right] \\
	&&+ \, 2\Xi^2\left\{ 4\mu^2\varphi \left( 1 - 3\varphi \right) + 12\varphi^2 + \mu^{4}\left[ 3\varphi \left( \varphi - 2\right) - 1\right] \right\}, \\
	c_{4} &=& 4\Xi \left\{ \left( \mu^2-1\right) \left[ \mu^2\left( \varphi - 1\right) - 4\varphi \right] - \Xi \left( \mu^{4} + \mu^2 - 2\right) \right\} \\
	&&- \, \mu^{4}\left( \varphi - 1\right)^2 + 8\varphi^2 - 4\mu^2\varphi \left( \varphi + 1\right), \\
	c_{6} &=& 4\left( 1 - \mu^2\right).
\end{eqnarray*}%

Equation~\eqref{eq:poly} is cubic in $\eta^2$ and can therefore be solved algebraically. To obtain the integration contour, we select the nonnegative root that satisfies \eqref{eq:ImagF=0} on the physical sheet and continuously traces the closed curve enclosing the branch cut. The upper and lower branches are denoted by $\eta_{\mu,\varphi,\tau_\sigma}^{\pm}(\xi)$, with $\eta_{\mu,\varphi,\tau_\sigma}^{+}(\xi)\geq0$ and
\begin{equation} \label{eq:eta+=-eta-}
	\eta_{\mu ,\varphi ,\tau_{\sigma}}^{+}\left( \xi \right) =
	-\eta_{\mu ,\varphi ,\tau_{\sigma}}^{-}\left( \xi \right),
\end{equation}
so that the SDP is symmetric with respect to the real axis.

We define
\begin{equation}
	f_{\delta}\left( x,t;s\right) = \frac{1}{2\pi i}\,
	\exp\!\left( s\left[t - \frac{x}{c}\,n\left( s\right) \right] \right)
\end{equation}
and justify the contour deformation before applying it. We use
\begin{equation}
	n(s) = \sqrt{\frac{s+\alpha}{s+\beta}},
	\qquad 
	\alpha > \beta > 0,
\end{equation}
such that $n(s) \to 1$ as $|s|\to\infty$. With this choice, $n(s)$ is single-valued and analytic in the cut plane
\begin{equation}
	\mathbb{C} \setminus \left[-\alpha,\,-\beta\right].
\end{equation}
Therefore, the integrand has two branch points, at $s=-\alpha$ and $s=-\beta$, but no isolated poles. Hence, no residue contribution arises when the Bromwich path is deformed into a contour surrounding the branch cut.

To control the contribution at infinity, we separate the wave-front term as in \eqref{eq:rdelta}--\eqref{eq:Af}. For large $|s|$,
\begin{equation}
	n(s) = 1 + \frac{\alpha-\beta}{2s} + \bigO{s^{-2}},
\end{equation}
and therefore
\begin{equation}
	\exp[-\chi s n(s)] = A_{\mathrm f}(x)e^{-\chi s} \left[1 + \bigO{s^{-1}}\right].
\end{equation}
Accordingly, let
\begin{equation}
	\widetilde r_{\delta,\mathrm{reg}}(x,s) = e^{-\chi s n(s)} - A_{\mathrm f}(x)e^{-\chi s}.
\end{equation}
It follows that
\begin{equation}
	\widetilde r_{\delta,\mathrm{reg}}(x,s) = \bigO{\frac{e^{-\chi s}}{s}}, \qquad |s|\to\infty.
\end{equation}
For $0<\mu<1$ (equivalently, $t > \chi$), the integrand of the inverse transform of the regular part consequently satisfies
\begin{equation}
	e^{st}\widetilde r_{\delta,\mathrm{reg}}(x,s) = \bigO{\frac{e^{(t-\chi)s}}{s}}.
\end{equation}
Thus, when the contour is closed in the left half-plane, the contribution from the contour at infinity vanishes as its radius tends to infinity. Since there are no isolated poles in the deformation region, the Bromwich contour can be deformed, without residue contributions, into a contour surrounding the branch cut and subsequently into the steepest-descent path $\gamma_\mu$.

The subtracted wave-front term is entire, so its integral around the closed contour vanishes:
\begin{equation}
	\int_{\gamma_\mu} A_{\mathrm f}(x)\, e^{(t-\chi)s}\,ds = 0.
\end{equation}
Consequently, although the wave-front contribution has to be separated
when discussing the behavior at infinity, the original integrand
$f_\delta\!\left(x, t; s\right)$ can still be used in the integral along
$\gamma_\mu$, which yields the regular part of the response.

As $\mu\to0^+$, the real saddles approach the branch points and $\gamma_\mu$ approaches its limiting configuration around the cut. At $\mu=0$, $F_0(s)=s$ has no finite saddle, so the contour is understood only as a limit. As $\mu\to1^-$, the two real saddles tend to $-\infty$ and $+\infty$, respectively. The wave-front Dirac term at $\mu=1$ is supplied separately by \eqref{eq:rdelta}.
For compactness, write $\eta(\xi)=\eta_{\mu,\varphi,\tau_\sigma}^{+}(\xi)$ and $s_\pm(\xi)=\xi\pm i\eta(\xi)$. The endpoints are $\xi_1=p_1<\xi_2=p_2$. We orient $\gamma_\mu$ counterclockwise, traversing the lower branch from $\xi_1$ to $\xi_2$ and the upper branch in the opposite direction. For $0<\mu<1$, equivalently $t>\chi$, this gives
\vskip -0.8truecm
\begin{equation} \label{eq:sdp_integral}
    \begin{split}
        r_{\delta,\mathrm{reg}}(x,t)
        &=\oint_{\gamma_\mu} f_\delta(x,t;s)\,ds\\
        &=\int_{\xi_1}^{\xi_2}\Bigl[
            f_\delta(x,t;s_-(\xi))(1-i\eta'(\xi)) -f_\delta(x,t;s_+(\xi))(1+i\eta'(\xi))
          \Bigr]\,d\xi\\
        &=-\frac{1}{\pi}\int_{\xi_1}^{\xi_2}
            \exp\!\left[t\,\realpart{F_\mu(s_+(\xi))}\right]
            \eta'(\xi)\,d\xi.
    \end{split}
\end{equation}
\vskip -0.25truecm
\noindent
The last equality uses conjugate symmetry and $F_\mu(s_-(\xi)) = F_\mu(s_+(\xi)) \in \mathbb{R}$ on the contour. The resulting real-line integral can be evaluated by adaptive Gauss--Kronrod quadrature.
Figures~\ref{fig:case1}--\ref{fig:case3} compare the regular response $r_{\delta,\mathrm{reg}}(x,t)$ obtained from \eqref{eq:integral_representation} and \eqref{eq:sdp_integral} for three values of $\tau_\epsilon$. The Dirac term is concentrated at $t=x/c$ and is not shown as a finite-valued curve. With $c=\tau_\sigma=1$, the wave front is the endpoint $x=t$ of each profile. The ordinate labelled $r_\delta(x,t)$ in the plots denotes the regular part only.
The corresponding Maxwell profiles, obtained in the limit $\tau_\epsilon\to\infty$ at fixed $c$ and $\tau_\sigma$, were presented in \cite{MainardiMentrelliSantander}.
\begin{figure}[H]
	\begin{center}
		\includegraphics[scale=0.4]{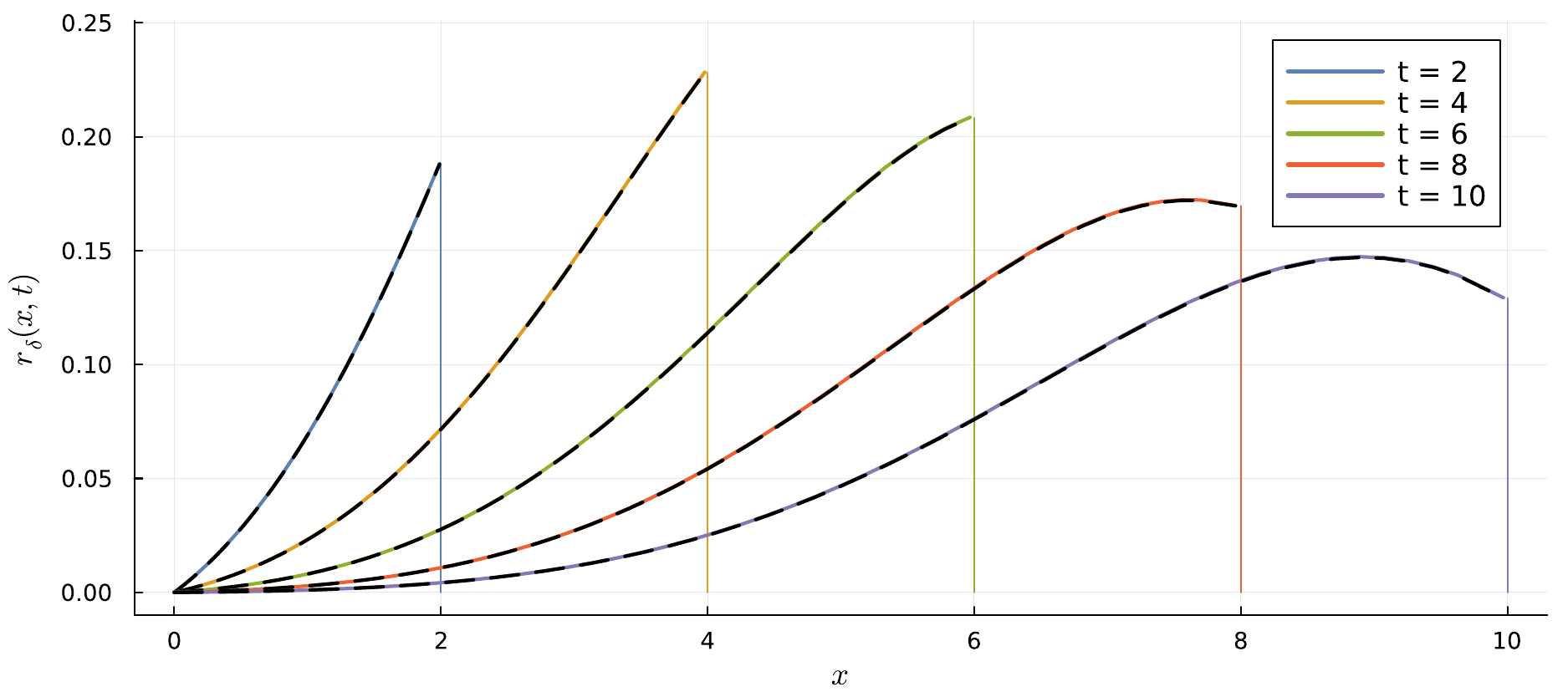}
		\includegraphics[scale=0.4]{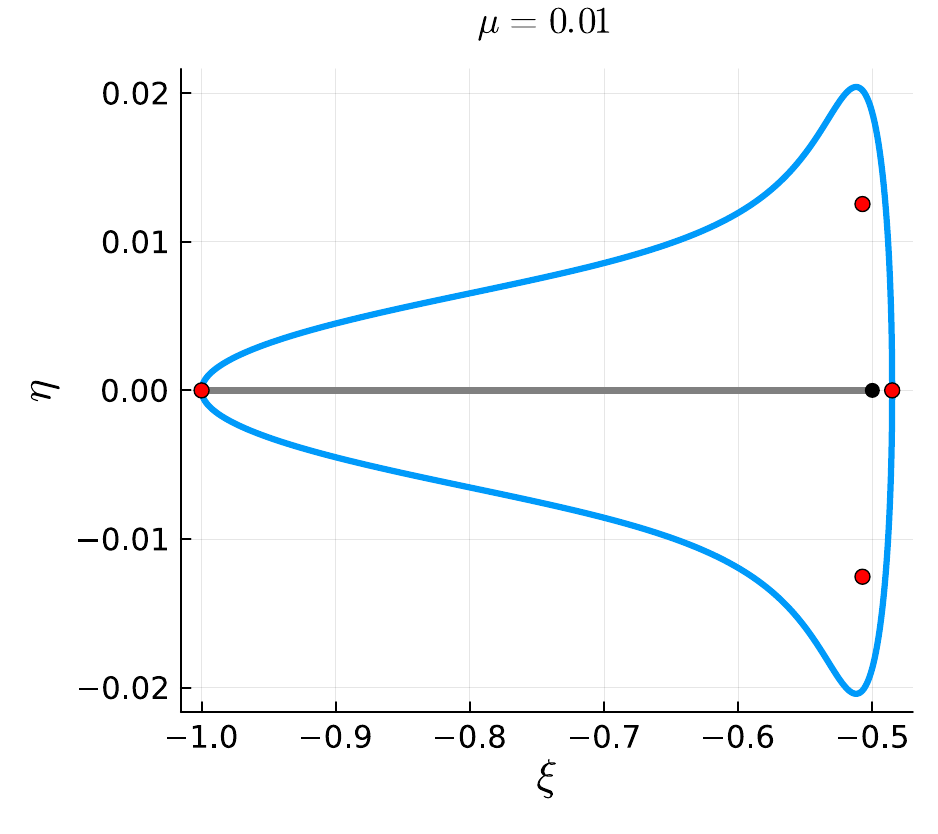} %
		\includegraphics[scale=0.4]{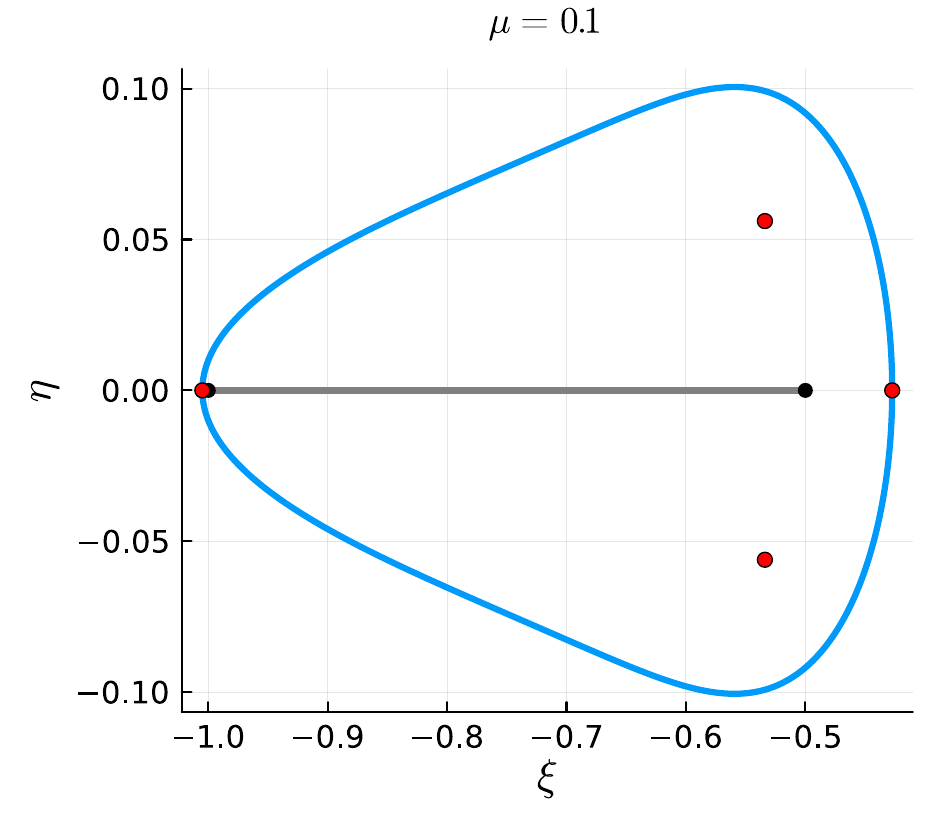} %
		\includegraphics[scale=0.4]{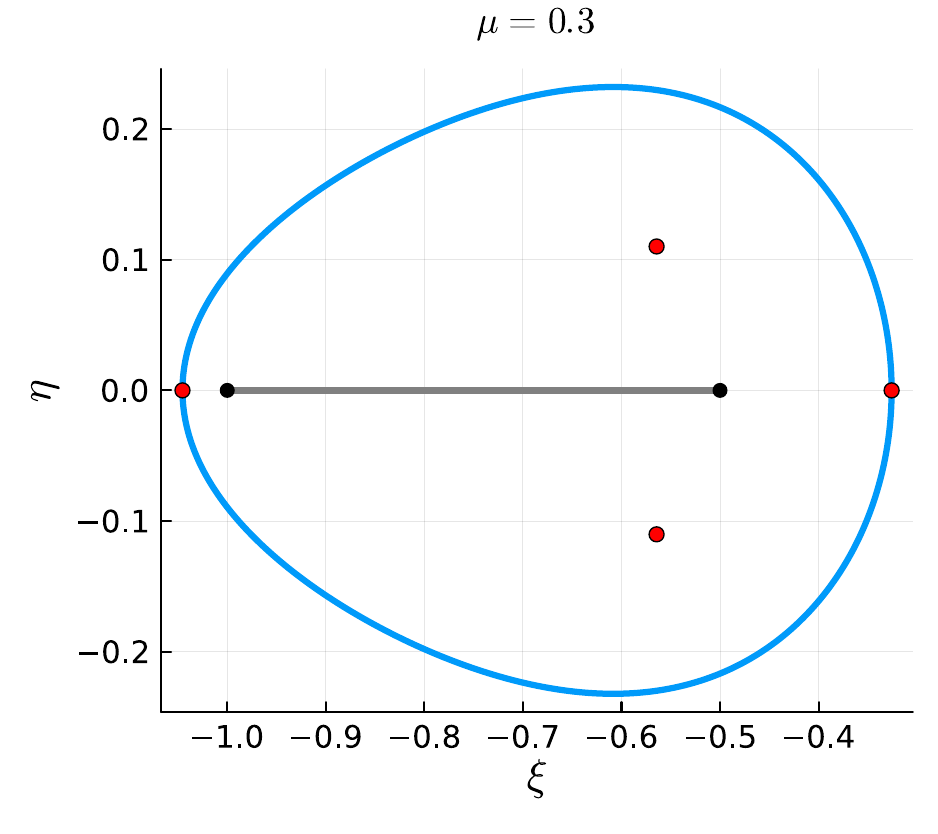} %
		\includegraphics[scale=0.4]{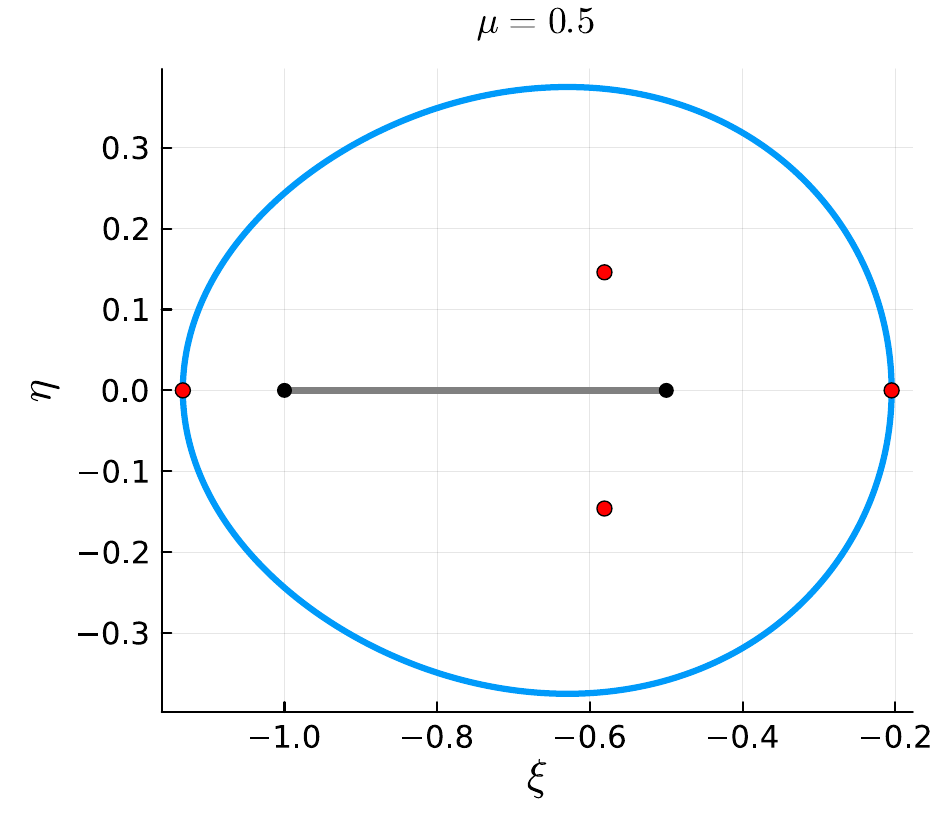} %
	\end{center}
	\vskip -0.85truecm
	\caption{Top: Regular part of the delta-pulse response computed from the integral representation (colored solid lines) and the steepest-descent contour integral (black dashed lines) at several times $t$, with $\tau_\sigma=1$, $\tau_\epsilon=2$, and $c=1$. The wave-front Dirac term is omitted. Bottom: steepest-descent paths for $\mu=x/(ct)=0.01,\ 0.1,\ 0.3,\ 0.5$.
	The black dots denote the two branch points. The red dots denote the four saddle points of the two-sheeted analytic continuation:	the two real points $p_1$ and $p_2$ lie on the physical sheet and on the integration contour $\gamma_\mu$, whereas the complex-conjugate points $p_3$ and $p_4$ belong to the second Riemann sheet and are shown only for reference; they do not contribute to the integral. \label{fig:case1}}
\end{figure}

\begin{figure}[H]
	\begin{center}
		\includegraphics[scale=0.4]{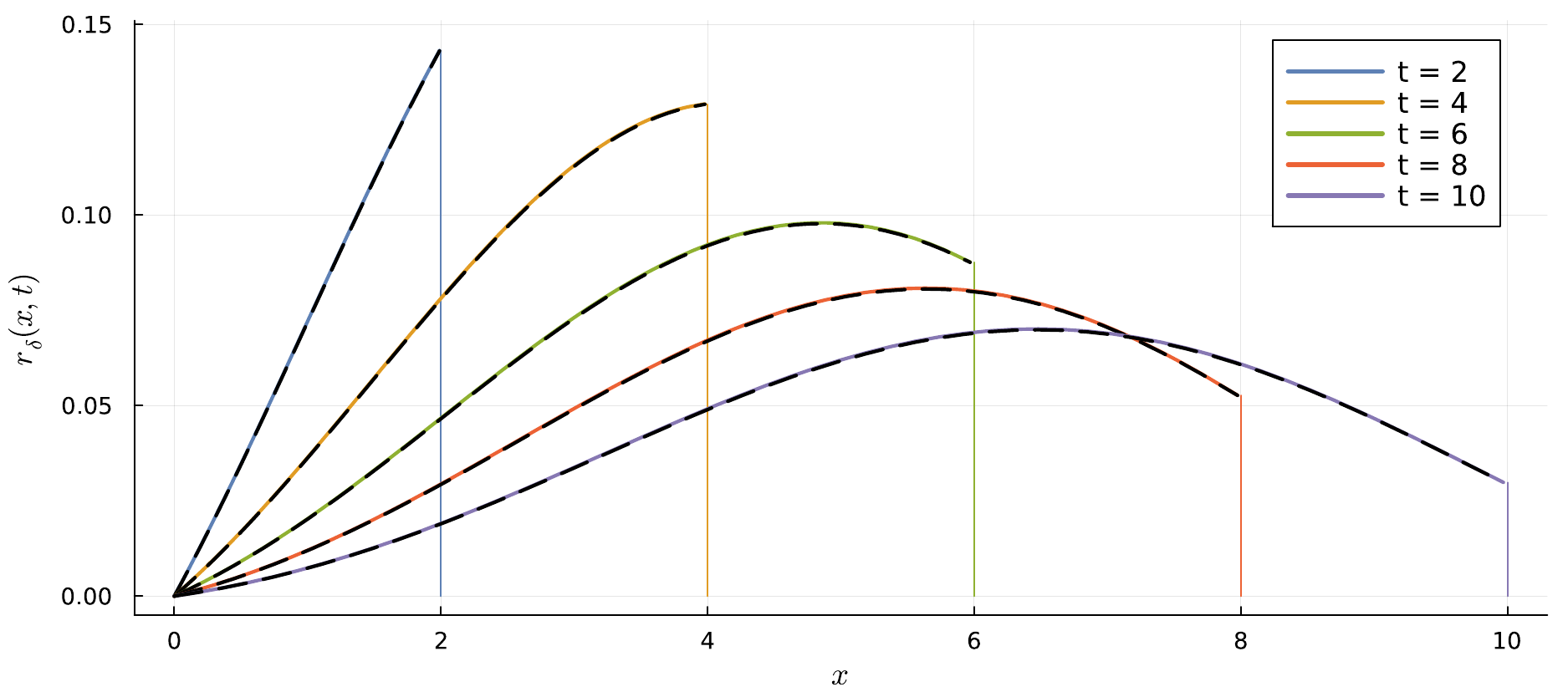}
		\includegraphics[scale=0.4]{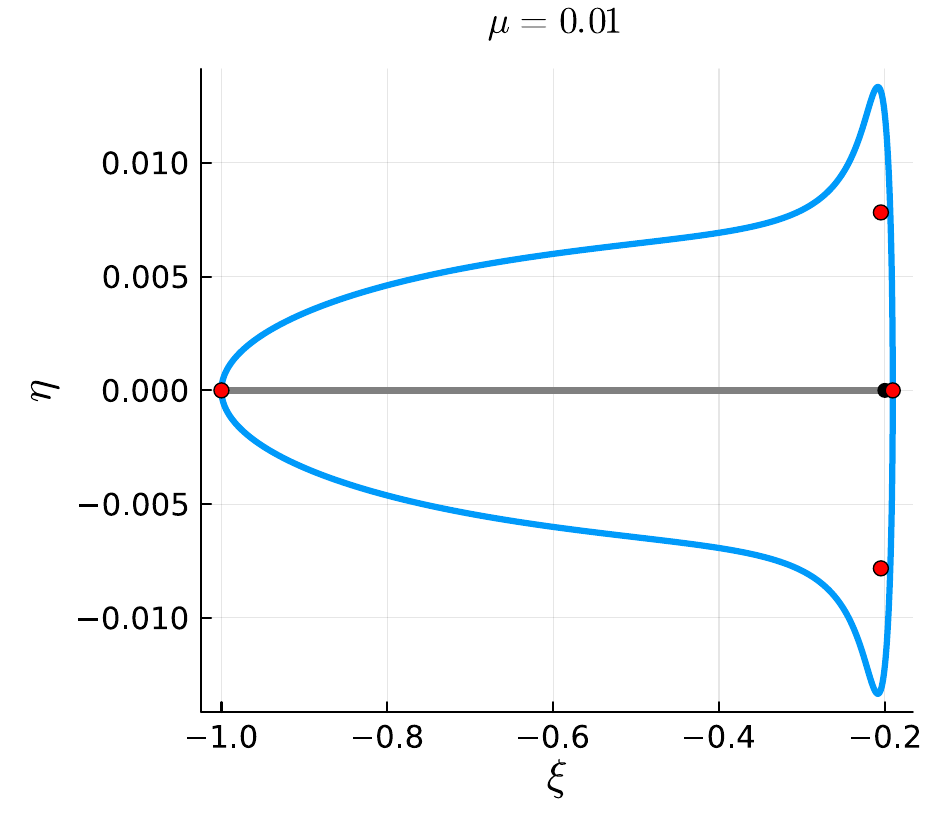} %
		\includegraphics[scale=0.4]{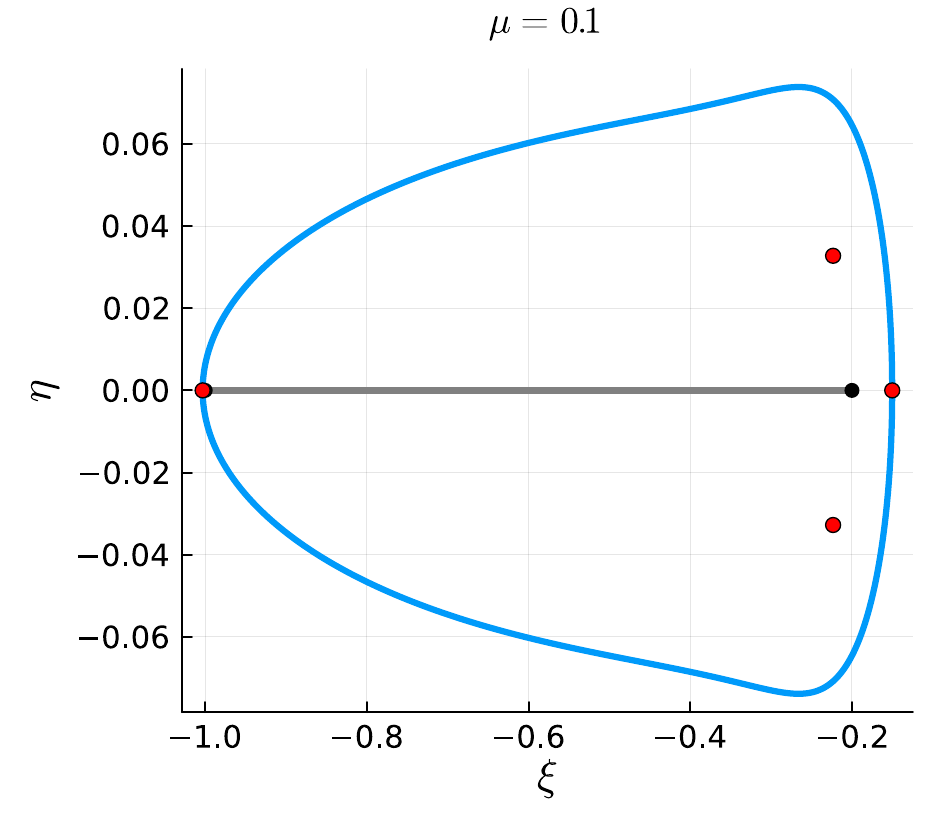} %
		\includegraphics[scale=0.4]{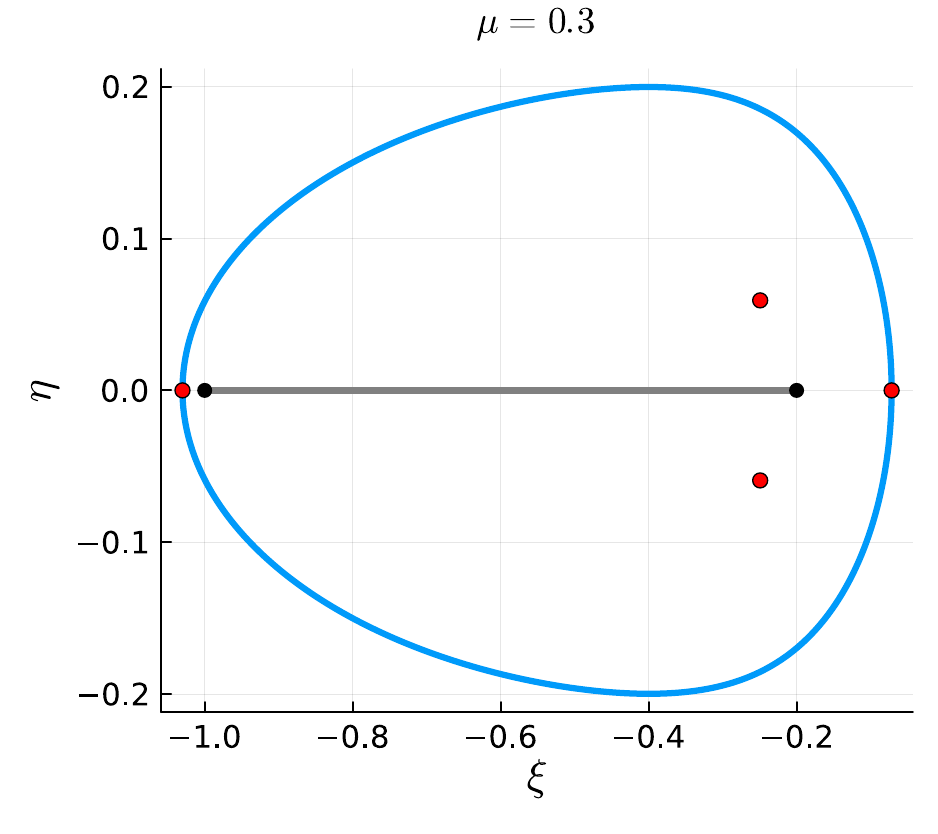} %
		\includegraphics[scale=0.4]{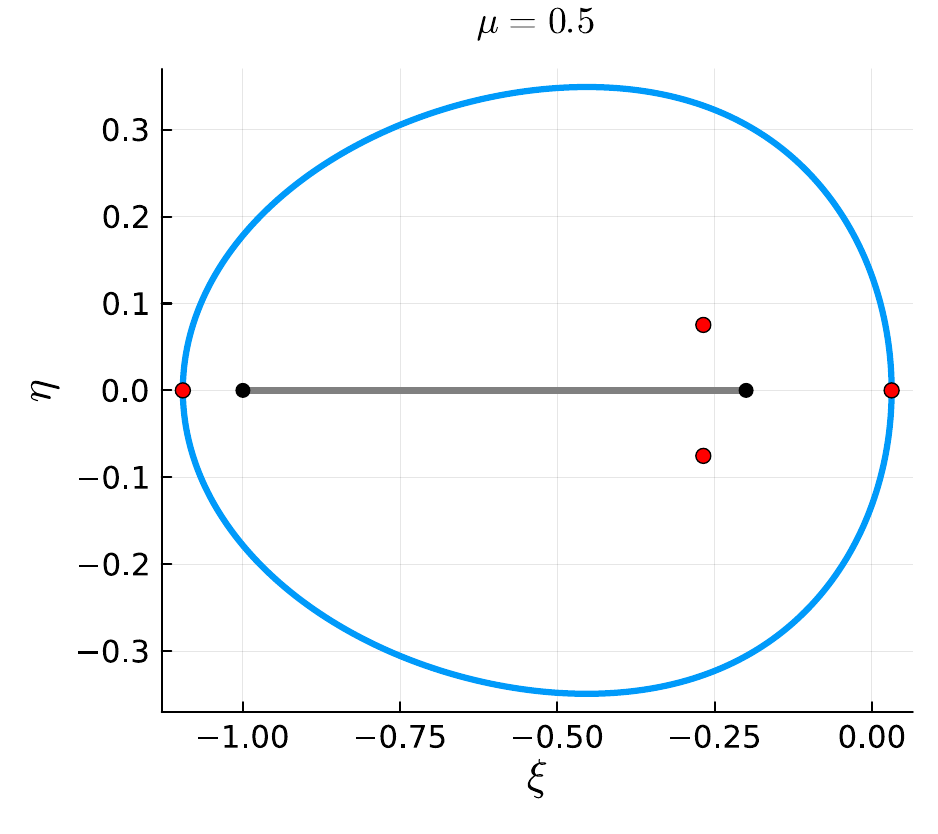} %
	\end{center}
	\vskip -0.85truecm
	\caption{Top: Regular part of the delta-pulse response computed from the integral representation (colored solid lines) and the steepest-descent contour integral (black dashed lines) at several times $t$, with $\tau_\sigma=1$, $\tau_\epsilon=5$, and $c=1$. The wave-front Dirac term is omitted. Bottom: steepest-descent paths for $\mu=x/(ct)=0.01,\ 0.1,\ 0.3,\ 0.5$.
	The black dots denote the two branch points. The red dots denote the four saddle points of the two-sheeted analytic continuation:	the two real points $p_1$ and $p_2$ lie on the physical sheet and on the integration contour $\gamma_\mu$, whereas the complex-conjugate points $p_3$ and $p_4$ belong to the second Riemann sheet and are shown only for reference; they do not contribute to the integral. \label{fig:case2}}
\end{figure}

\begin{figure}[H]
	\begin{center}
		\includegraphics[scale=0.4]{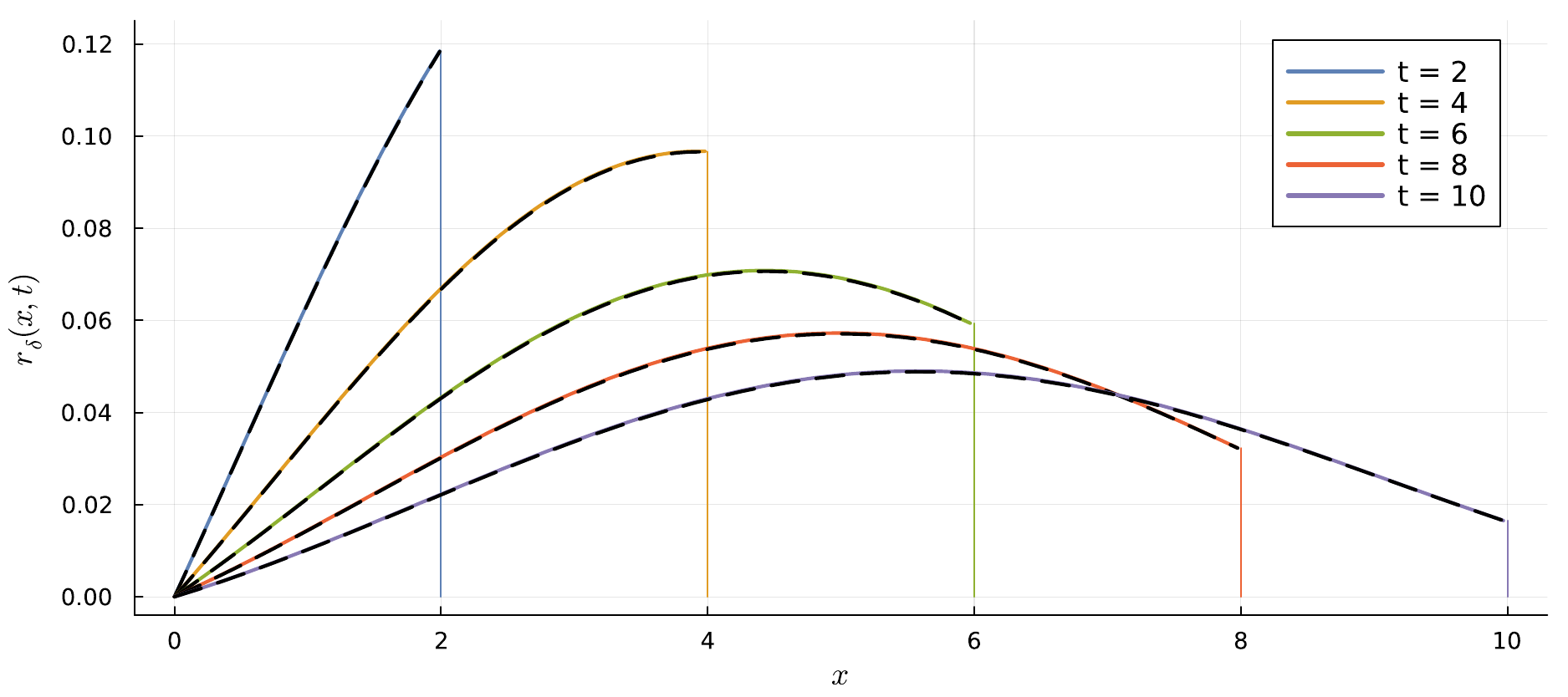}
		\includegraphics[scale=0.4]{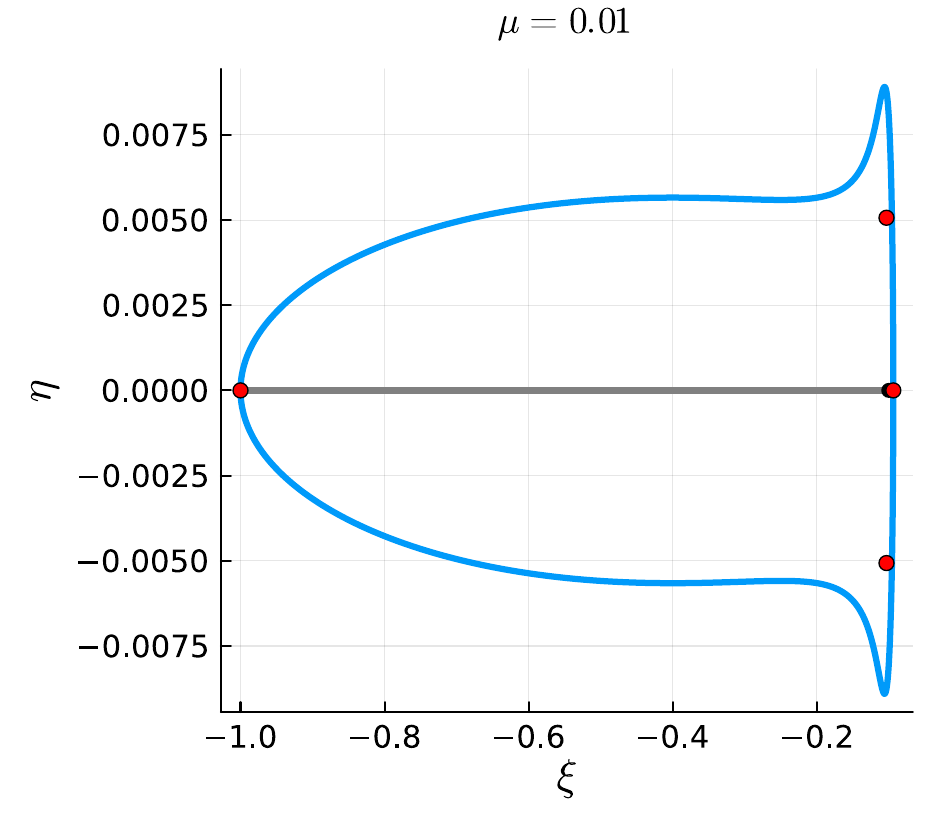} %
		\includegraphics[scale=0.4]{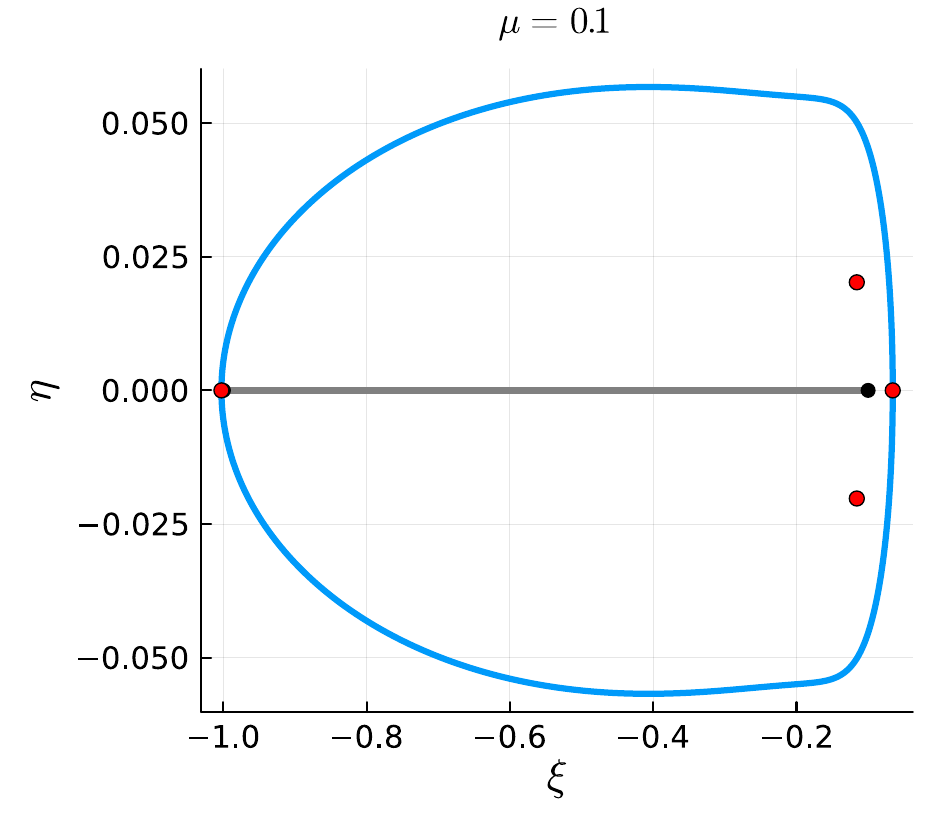} %
		\includegraphics[scale=0.4]{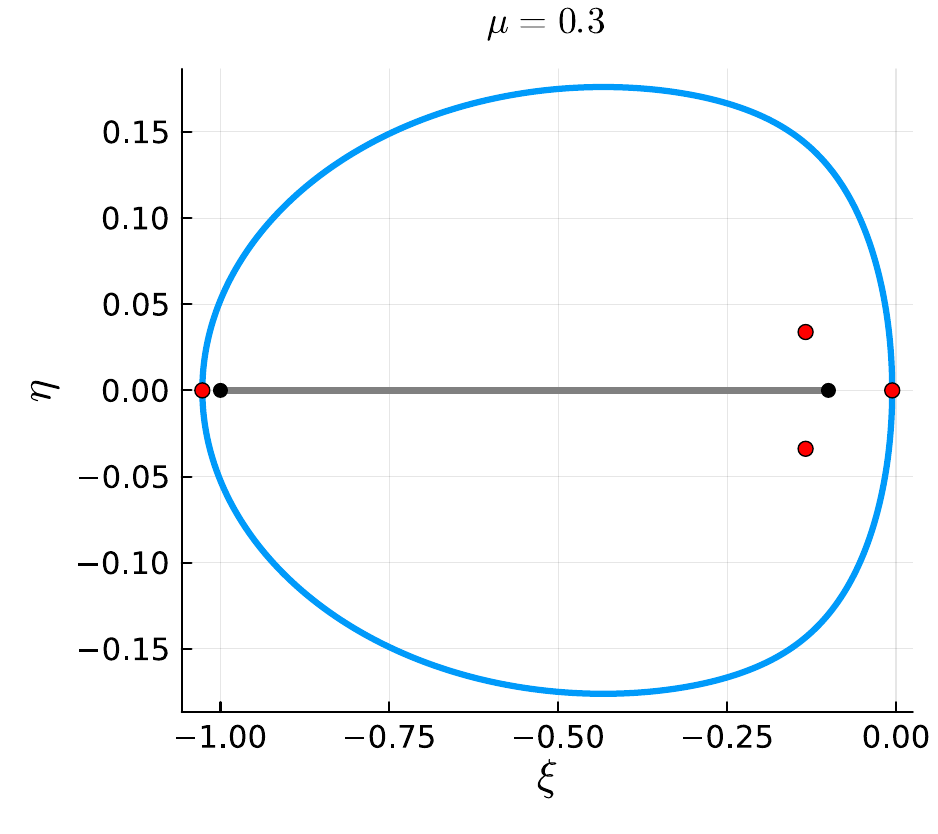} %
		\includegraphics[scale=0.4]{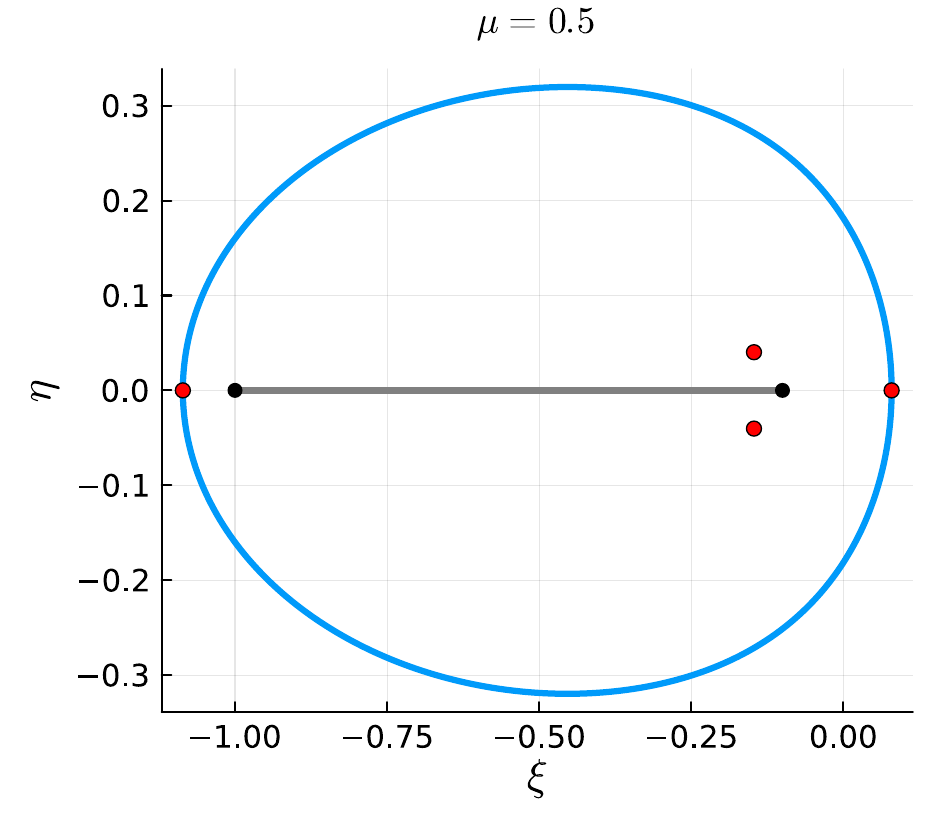} %
	\end{center}
	\vskip -0.85truecm
	\caption{Top: Regular part of the delta-pulse response computed from the integral representation (colored solid lines) and the steepest-descent contour integral (black dashed lines) at several times $t$, with $\tau_\sigma=1$, $\tau_\epsilon=10$, and $c=1$. The wave-front Dirac term is omitted. Bottom: steepest-descent paths for $\mu=x/(ct)=0.01,\ 0.1,\ 0.3,\ 0.5$.
	The black dots denote the two branch points. The red dots denote the four saddle points of the two-sheeted analytic continuation:	the two real points $p_1$ and $p_2$ lie on the physical sheet and on the integration contour $\gamma_\mu$, whereas the complex-conjugate points $p_3$ and $p_4$ belong to the second Riemann sheet and are shown only for reference; they do not contribute to the integral. \label{fig:case3}}
\end{figure}

\section{Conclusions} \label{Section 6: Conclusions}

We have studied the response to an impulsive boundary excitation in a semi-infinite, homogeneous, linear viscoelastic medium governed by the Zener model. This type I model has a finite wave-front velocity. We developed and compared two complementary approaches.

The first approach gives an exact integral representation by analytical inversion and convolution. It involves modified Bessel functions of the first kind and Macdonald functions of half-integer order; the latter yield an equivalent series with elementary coefficients through Remark~\ref{Remark1}. The representation reduces to the known Maxwell response as $\tau_\epsilon\to\infty$ at fixed relaxation time and instantaneous modulus.

The second approach constructs the SDP of the Zener phase function. The contour satisfies a sixth-degree polynomial equation in the imaginary coordinate $\eta$, which is cubic in $\eta^2$. Selecting the appropriate root gives the upper and lower contour branches. We characterized the four saddles of the two-sheeted phase for $0<\mu<1$ and reduced the contour integral to the real-line integral \eqref{eq:sdp_integral}, suitable for adaptive Gauss--Kronrod quadrature.

The regular profiles in Figs.~\ref{fig:case1}--\ref{fig:case3}, computed for three values of $\tau_\epsilon$ with $\tau_\sigma=c=1$, show close agreement between the two approaches. The plotted profiles were also reproduced by applying Talbot's numerical inversion method \cite{NumericalLaplace} to the transform in \eqref{r_delta=L-1[n(s)]}, with the index \eqref{eq:SLS n(s) (4.51)}, at points behind the wave front.

The Zener response involves two characteristic velocities. The leading edge of the signal propagates with the finite wave-front velocity
\begin{equation}
	c=\sqrt{\frac{G_g}{\rho}},	
\end{equation}
whereas the low-frequency phase velocity tends to the equilibrium velocity
\begin{equation}
	c_0 = \sqrt{\frac{G_e}{\rho}} = c\sqrt{\frac{\tau_\sigma}{\tau_\epsilon}}.
\end{equation}
The frequency dependence of the phase velocity, from $c_0$ at low frequencies to $c$ at high frequencies, produces dispersion and a progressive deformation of the initially localized pulse. The response vanishes for $x>ct$, so $c$ is the signal velocity. This distinction is relevant to seismic pulse arrivals: a sharp high-frequency onset may precede the lower-frequency part of the signal.

At fixed $c$ and $\tau_\sigma$, the ratio $R = \tau_\epsilon/\tau_\sigma>1$ quantifies the transition from the Zener response to the Maxwell response. Indeed, Eq.~\eqref{eq:Af} shows that the singular wave-front amplitude decays with distance as
\begin{equation}
	A_{\mathrm f}(x) = 
	\exp\left[-\frac{x}{2c\tau_\sigma}
	\left(1-\frac{1}{R}\right)\right].	
\end{equation}
For the values $R = 2, 5, 10$ used in Figs.~\ref{fig:case1}--\ref{fig:case3}, with $c = \tau_\sigma = 1$, the corresponding spatial attenuation coefficients
are $0.25$, $0.4$, and $0.45$, respectively, compared with $0.5$ in the Maxwell limit $R \to \infty$. Therefore, increasing $R$ lowers the equilibrium velocity $c_0 = c/\sqrt{R}$, strengthens the attenuation of the wave front, and modifies the shape and spreading of the regular wake displayed in the figures. In the opposite limit $R\to1$, one has $c_0\to c$, the attenuation coefficient tends to zero, and the nondispersive elastic response is recovered.

The contour method provides an independent numerical evaluation, while the integral representation gives analytical information on the response and its parameter dependence.

\smallskip

The absence of a closed-form inverse Laplace transform needs not prevents the use of the SDP approach in more general viscoelastic models. Given $n(s)$ in the Laplace domain, the relevant saddles can be located numerically from
\begin{equation}
	\frac{dF_{\mu }(s)}{ds} = 0.
\end{equation}
For a nondegenerate saddle, the local steepest-descent directions are determined by $d^2F_\mu/ds^2$. Starting slightly away from the saddle along these directions, one can trace the contour numerically and evaluate the inverse transform by quadrature. The contour must be chosen to reproduce the Bromwich integral, with any required residue, branch-cut, or wave-front contributions included separately. The explicit polynomial description available for the Zener model simplifies this construction but is not essential to it.

These results extend the Maxwell analysis in \cite{MainardiMentrelliSantander} and provide a basis for studying impulsive waves in more general viscoelastic media, including fractional constitutive models.

\newpage

\subsection*{Acknowledgements}

The work of A. M. and F. M. was carried out in the framework of the activities of the Italian National Group for Mathematical Physics (GNFM/INdAM). Furthermore, A. M. acknowledges the Italian National Institute for Nuclear Physics (INFN), FLAG grant, for partial support.
The authors are grateful to the anonymous
referees and the editor of Journal of Elasticity for valuable comments and suggestions, which have helped  improve the clarity and  presentation of the 
manuscript.

\medskip

\subsection*{AI Disclosure} 
OpenAI's GPT--6 (free version) was used during November 2026 for critical feedback, bibliographic assistance and language improvement. The proof strategies were conceived by the authors and numerical results were obtained without AI assistance. All AI-assisted content was independently checked for correctness by the authors, who retain full responsibility for the manuscript.

\end{document}